\documentclass[conference,a4paper]{IEEEtran}

\IEEEoverridecommandlockouts
\usepackage{cite}
\usepackage{amssymb,amsfonts}
\usepackage{graphicx}
\usepackage{textcomp}
\usepackage{xcolor}
\usepackage{enumitem}
\usepackage{tikz}
\usepackage[ruled,vlined,linesnumbered]{algorithm2e}
\usepackage[cmex10]{amsmath}
\usepackage{array}
\usepackage{amsthm}
\usepackage{mathptmx,stmaryrd}
\usepackage{subfigure}
\usepackage{hyperref}
\usepackage{pgfplots}
\pgfplotsset{compat=1.10}

\usepgfplotslibrary{fillbetween}

\newcommand{\eps}{\ensuremath{\varepsilon}}

\newcommand{\bea}{\begin{eqnarray}}
\newcommand{\eea}{\end{eqnarray}}

\newcommand{\bali}{\begin{align}} 
\newcommand{\eali}{\end{align}}

\newcommand{\bean}{\begin{eqnarray*}} 
\newcommand{\eean}{\end{eqnarray*}}

\newcommand{\beq}{\begin{equation}}
\newcommand{\eeq}{\end{equation}}

\newcommand{\mcZ}{\ensuremath{\mathcal{Z}}}

\newcommand{\mcD}{\ensuremath{\mathcal{D}}}
\newcommand{\mcM}{\ensuremath{\mathcal{M}}}
\newcommand{\mcC}{\ensuremath{\mathcal{C}}}
\newcommand{\mcR}{\ensuremath{\mathcal{R}}}

\newcommand{\mcO}{\ensuremath{\mathcal{O}}}
\newcommand{\mcQ}{\ensuremath{\mathcal{Q}}}
\newcommand{\mcX}{\ensuremath{\mathcal{X}}}
\newcommand{\mcG}{\ensuremath{\mathcal{G}}}
\newcommand{\mcE}{\ensuremath{\mathcal{E}}}

\newcommand{\mcY}{\ensuremath{\mathcal{Y}}}

\newcommand{\RR}{\ensuremath{\mathbb{R}}}

\newcommand{\lra}{\ensuremath{\leftrightarrow}}

\newcommand{\lp}{\ensuremath{\left (}}
\newcommand{\rp}{\ensuremath{\right )}}

\newcommand{\llb}{\ensuremath{\llbracket}} 
\newcommand{\rrb}{\ensuremath{\rrbracket}}

\usetikzlibrary{calc,shapes,arrows,positioning}
\def\BibTeX{{\rm B\kern-.05em{\sc i\kern-.025em b}\kern-.08em
	T\kern-.1667em\lower.7ex\hbox{E}\kern-.125emX}}
\newtheorem{theorem}{Theorem}
\newtheorem{lemma}{Lemma}
\newtheorem{remark}{Remark}
\newtheorem*{example}{Example}
\newtheorem{definition}{Definition}

\tikzstyle{block} = [draw, fill=white, rectangle, minimum height=4em, minimum width=3em]
\tikzstyle{adder} = [draw, fill=white, circle, path picture={
	\draw[black] 
	(path picture bounding box.south) -- (path picture bounding box.north)
	(path picture bounding box.west) -- (path picture bounding box.east);
	}]
\tikzstyle{input} = [coordinate]
\tikzstyle{output} = [coordinate]
\tikzstyle{pinstyle} = [pin edge={to-,thin,black}]
\definecolor{TUMpantone540c}{RGB}{000,051,089}
\definecolor{TUMpantone301c}{RGB}{000,082,147}
\definecolor{TUMpantone285c}{RGB}{000,115,207}
\definecolor{TUMpantone542c}{RGB}{100,160,200}
\definecolor{TUMpantone283c}{RGB}{152,198,234}
\definecolor{TUMdgray}{RGB}{088,088,090}
\definecolor{TUMmgray}{RGB}{156,157,159}
\definecolor{TUMmgray}{RGB}{215,217,218}
\definecolor{TUMyellow}{RGB}{255,180,000}
\definecolor{TUMorange}{RGB}{255,128,000}
\definecolor{TUMblue}{RGB}{0,101,189}
\definecolor{TUMpantone301}{RGB}{0,82,147}
\definecolor{TUMgreen}{RGB}{0,124,48}
\definecolor{TUMred}{RGB}{196,7,27}
\definecolor{green4}{rgb}{0,0.57,0}
\definecolor{blue4}{rgb}{0,0,0.57}
\definecolor{lavender2}{rgb}{0.9,0.9,0.98}
\definecolor{lightblue}{rgb}{0.74,0.93,1}
\definecolor{white}{rgb}{1,1,1}
\definecolor{lightyellow}{rgb}{0.93,0.93,0.82}
\definecolor{plum}{rgb}{0.86,0.62,0.86}
\definecolor{palegreen}{rgb}{0.6,1,0.6}
\definecolor{markerred}{rgb}{1,0.85,0.85}
\definecolor{hellgrau}{rgb}{0.95,0.95,0.95}
\tikzstyle{block} = [draw, fill=white, rectangle, minimum height=4em, minimum width=3em]
\tikzstyle{input} = [coordinate]
\tikzstyle{output} = [coordinate]
\tikzstyle{pinstyle} = [pin edge={to-,thin,black}]
\tikzstyle{block} = [draw, fill=white, rectangle, minimum height=4em, minimum width=3em]
\tikzstyle{EllipBlock} = [draw, ellipse, minimum height=4em, minimum width=3em]
\tikzset{
	cloud node/.style={
		cloud, cloud puff arc=140, inner sep=0pt, minimum width=3.5cm,
		minimum height=2.5cm, draw}}
\tikzstyle{input} = [coordinate]
\tikzstyle{output} = [coordinate]
\tikzstyle{pinstyle} = [pin edge={to-,thin,black}]
\tikzstyle{arrow} = [->,>=stealth]

\begin{document}

\title{Uniformly Bounded State Estimation over \\
	Multiple Access Channels \\
\thanks{Corresponding author G.N. Nair. This work was supported by the Australian Research Council under Future Fellowship grant FT140100527. Parts of this work have been presented at ISIT 2020 \cite{zafzoufISIT2020}.}
} 

\author{\IEEEauthorblockN{Ghassen Zafzouf, Girish N.~Nair and Farhad Farokhi}
\IEEEauthorblockA{\textit{Department of Electrical and Electronic Engineering}\\
\textit{University of Melbourne}\\
VIC 3010, Australia\\
gzafzouf@student.unimelb.edu.au, $\lbrace$gnair, ffarokhi$\rbrace$@unimelb.edu.au 
}
}

\maketitle

\thispagestyle{plain}
\pagestyle{plain}

\begin{abstract}
This paper addresses the problem of distributed state estimation via \textit{multiple access channels} (MACs). We consider a scenario where two encoders are simultaneously communicating their measurements through a noisy channel. Firstly, the zero-error capacity region of the general $M$-input, single-output MAC is characterized using tools from nonstochastic information theory. Next, we show that a tight condition to be able to achieve uniformly bounded state estimation errors can be given in terms of the channel zero-error capacity region. This criterion relates the channel properties to the plant dynamics. These results pave the way towards understanding information flows in networked control systems with multiple transmitters.   
\end{abstract}

\section{Introduction}
Over the last few decades, the important progress in the fields of electronics, nanotechnology and computing has led to the development of smaller and cheaper sensors that can perform real-time computation. The availability of such sensors to a wide range of users has contributed, among other factors, to a tremendous increase in the number of connected devices, i.e., large wireless sensor networks (WSNs) \cite{marano2008, chisci2019, conti2019, zabini2016}, and to the emergence of new technological concepts such as the \textit{Internet of Things} (IoT) and \textit{Machine-to-Machine} (M2M) communication \cite{shirvanimoghaddam2018}. 

Traditionally, classical control and estimation theory has assumed that communication between different components of a network occurs over point-to-point links. This standard assumption however cannot hold for recent emerging applications, where numerous subsystems are interacting with each other. In these applications, it may be impractical or costly to set up  and maintain multiple physical point-to-point links between multiple subsystems. One solution is to use a shared medium such as wireless, with frequency and/or time divided up into non-interfering slots, each for the dedicated use by a single transmitter \cite{li2010}. In the case of time-slots, this leads to deterministic round-robin-like protocols for channel access \cite{nesic2007}. However, such methods are generally suboptimal since they do not fully exploit channel resources in order to increase throughput. For real-time applications, the delays caused by a round-robin scheduling protocol may also degrade performance significantly. Event-based time scheduling is an alternative that allows sensors to transmit more quickly when their measurements exceed certain levels or increase rapidly \cite{heemels2010}. This allows communication resources to be used where they are needed most. However, an additional network layer is typically required to make sure network access goes to the device that most needs it.

An alternative is to allow  users to transmit their messages simultaneously over the shared channel, and use coding and modulation to mitigate inter-user interference in addition to noise and other channel effects. Since resources such as specific time slots or frequency bands are not pre-allocated to each user, this method potentially allows faster data flows to be achieved. In real-time applications with a large number of transmitters, it could also improve latency compared to a time-division approach. A simple model of simultaneous communication is the {\em multiple access channel (MAC)}, consisting of different users who aim to each send an independent message reliably to a common receiver. This was initially introduced by Shannon in his seminal work \cite{b3}. In Part \ref{subsec:MAC} of this introduction, we present an example of a noisy MAC to familiarize the reader with this model.

In contrast to communication systems, where performance is studied in an average sense, many control and estimation systems are deployed for mission-critical purposes, with any error leading to potentially devastating consequences. 
This motivates system designs based on the worst-case scenario \cite{nair2013, nair2015, massimo2019, gargrani2020}. 
With this in mind, the aim of this work is to understand the problem of remote estimation of dynamical systems over MACs with bounded noise rather than assuming a statistical distribution. In this context, we try to answer the following fundamental questions: is it possible to provide a reliable estimation of the states of distinct plants observed by different sensors whose measurements are sent simultaneously over a wireless channel? If so, what is the connection between the intrinsic properties of the dynamical systems and the transmission data rates?     

\subsection{Multiple Access Channels (MACs) \& Zero-Error Capacity}
\label{subsec:MAC}
As discussed above, allowing senders to transmit simultaneously can be advantageous in terms of improving throughput and reducing latency. A channel model capturing the essence of this problem is the aforementioned MAC. A natural example is the Gaussian two-user MAC
\begin{align}
	Y = X^1 + X^2 + Z, 
\end{align} 
where $X^1$ and $X^2$ are the inputs to the channel, $Z$ denotes zero-mean Gaussian noise independent of $X^1$ and $X^2$, and $Y$ is the channel output. In \cite{avestimehr2011}, it was shown that a reasonably good approximation of the Gaussian MAC can be obtained by approximating it as a deterministic two-user XOR MAC model carrying summation in $\mathbb{Z}_2$. A further example of MACs is the {\em binary adder channel}, which unlike the XOR model, performs the addition over $\mathbb{Z}$ and will be discussed in greater detail in Section \ref{sec:BAC}.  

In communications, the (ordinary) \textit{channel capacity} $C$ is the maximum rate at which information can be transferred with vanishingly small error probability across the channel \cite{coverThomas91}. Unlike point-to-point channels, whose channel capacity $C$ is a non-negative scalar, in the context of multi-user channels, including the MAC in particular, we talk about a \textit{channel capacity region} $\mcC$ that lies in an $M$-dimensional space, where $M$ is the number of users in the system.  
For more details regarding this topic we refer the reader to \cite{gamal12}.

In the context of worst-case state estimation, it turns out that the notion of \textit{zero-error capacity} $C_0$ is a more insightful figure of merit than classical channel capacity $C$ \cite{nair2013, nair2015}. The zero-error capacity $C_0$ of a point-to-point channel is defined as the highest block-coding rate which yields exactly zero decoding errors at the receiver \cite{ZEshannon1956}. Intuitively, the block-codes with zero-error property are those leading to an output set that consists of distinct elements, i.e., no codewords result in the same output. 

In a similar manner to $\mcC$, in a multi-user communication setup, the \textit{zero-error capacity region} $\mcC_0$ with $M$ senders is contained in an $M$-dimensional space, rather than a single axis like $C_0$. A formal definition of $\mcC_0$ is introduced in Section \ref{sec:ZECapacity}.      

\subsection{Literature Overview}
The decisive role of information theory in digital communications arises largely from the fundamental coding bounds it provides. To approach these bounds, coding schemes with arbitrarily long block lengths are used, which result in vanishingly small probability of decoding errors. However, long block-lengths are ill-suited for real-time control systems, since they lead to long delays that degrade closed-loop performance significantly. Furthermore, in real-time control and estimation applications where there are safety requirements or mission-critical objectives, closed-loop performance is often quantified in a worst-case sense, rather than probabilistically.  

An important step towards understanding communication requirements for worst-case state estimation was taken in \cite{matveevSavkin2007}, where it was shown that to achieve almost surely (a.s.) uniform state estimation of a linear, disturbed dynamical system over a stochastic discrete memoryless point-to-point channel, it is necessary and the open-loop \emph{topological entropy}\footnote{The notion of topological entropy of a system is defined as the sum-log of its unstable pole magnitudes.} $h$ does not exceed the channel zero-error capacity $C_0$. If $C_0>h$ strictly, then there exists a coding and estimation scheme that achieves a.s. uniform state estimation. By noting that $C_0$ depends on the combinatorial structure  of the channel rather than its probabilistic nature, this result was rederived in \cite{nair2013} for surely bounded estimation, by introducing the framework of \emph{uncertain variables} (uv's) and \emph{nonstochastic information} $I_*$. In a recent article \cite{wiese2019}, the authors introduced the notion of uncertain wiretap channel and studied the problem of secure state estimation of a noisy unstable dynamical system in a nonstochastic setup. Tools from this non-probabilistic framework were also used to study the problem of bounded state estimation over point-to-point channels with finite memory \cite{saberi2020, saberi2021}.

The main focus in previous works was on studying either state estimation or control problem in the context of point-to-point channels \cite{sahai2006, nair2007, borkar2001,nigel2010}. 
Stochastic stabilization of networked control systems over multiple wireless channels has been extensively studied in, e.g., \cite{rich2016, garone2012}.
Furthermore, the problem of filtering for stochastic systems consisting of distributed sensors exchanging information over point-to-point links was investigated for both discrete, e.g., in \cite{saber2007, saber2009}, and continuous settings, e.g., in \cite{tanwani2021}.
To the best of our knowledge, the problem of state estimation over a MAC with bounded noise has not been considered yet. An exception is the work of Zaidi \textit{et al.} \cite{zaidi2010}, where classical stochastic tools are used to present sufficient conditions ensuring the mean-square stabilization of two scalar linear time-invariant (LTI) systems over a noisy two-input, single output MAC. The paper \cite{gupta2019} also obtains necessary and sufficient conditions for stabilizing two scalar plants across a shared Gaussian MAC. In that paper, the authors distinguish between the case where encoders are entirely independent from each other, and where information sharing among them is allowed. 

In contrast, we address the problem of worst-case state estimation over MAC with bounded noise without assuming a specific statistical distribution. Hence, with the aim of understanding distributed systems with multiple sensors sharing the same channel, the characterization of the zero-error capacity region $\mcC_0$ of an $M$-input, single-output MAC with bounded noise becomes necessary. To date, no formula for such a region has been obtained and $\mcC_0$ of MACs in general remains an open problem, in contrast to its ordinary capacity region studied in, e.g., \cite{ahlswede71,liao72}. The work in \cite{nairITW2019} characterized the zero-error capacity region for the special case of a MAC with two correlated transmitters. Another class of MACs has been studied in \cite{zafzoufITW2020}, namely MAC with pairwise shared messages as well as a common message among all users. To derive an expression of the zero-error capacity region, we use \textit{nonstochastic information theory} \cite{nair2013, nair2015, massimo2019}. However, unlike the point-to point zero-error capacity in \cite{nair2015}, our definition requires the concept of \textit{conditional nonstochastic information}.        
\subsection{Overview of Main Result}
Our main contributions are hence the following: \textbf{(1)} We firstly generalize the characterization of the zero-error capacity region obtained in \cite{nairITW2019} for a two-user MAC to the $M$-user case with one common message using the framework of non-stochastic information (Theorem \ref{conjecture:DMAC}). To this end, we present a converse proof and establish the achievability argument by constructing a suitable coding scheme. \textbf{(2)} Next, the problem of state estimation over MACs with bounded noise is studied and a theorem (Theorem \ref{th:StateEstimation}) linking between the intrinsic properties of the dynamical systems and the zero-error capacity of the communication channel is shown. A numerical example involving the binary adder channel is then used to illustrate this theoretical result.   

The first step towards understanding information flow in large networked systems with multiple sensors and estimators is to start with analyzing a simple network topology. Consider the setup depicted in Fig.~\ref{fig:setup}. It consists of three discrete LTI dynamical systems characterized by the following system equations for $i \in \lbrace 0,1,2 \rbrace $:
\begin{subequations} 
	\begin{align}
		X^i(k+1) &= A^i X^i(k) + V^i(k) \in \mathbb{R}^{d_i}, \label{originalState1} \\
		Y^i(k)   &= C^i X^i(k) + W^i(k) \in \mathbb{R}^{b_i}, \label{originalOutput1}  
	\end{align}
\end{subequations}
where $V^i(k)$ and $W^i(k)$ denote process and measurement noise signals at time instant $k \in \mathbb{Z}_{\geq 0}$.
\begin{figure}
\centering
\scalebox{.75}{
	\begin{tikzpicture}[auto, node distance=2cm,>=latex'] [scale=1]
	\node[block, name=P1lant] [minimum width=13mm, minimum height=10mm, pin={[pinstyle]above:{\small $w^{1}(k), v^{1}(k)$}}] (P1lant) {Plant 1};
	\node[inner sep=0,minimum size=0, right = 0cm of P1lant] (k12) {};
	\node[block, below =1cm of P1lant] [minimum width=13mm, minimum height=10mm, pin={[pinstyle]above:{\small $w^{0}(k), v^{0}(k)$}}] (plant0) {Plant 0}; 
	\node[circle,fill,inner sep=.7pt, right =  1.37cm of plant0] (k0) {};    
	\node[block, below =1cm of plant0] [minimum width=13mm, minimum height=10mm, pin={[pinstyle]above:{\small $w^{2}(k), v^{2}(k)$}}] (P2lant) {Plant 2}; 
	\node[inner sep=0,minimum size=0, right = 0cm of P2lant] (k21) {};
	\node[block, right =1cm of k12] [minimum width=8mm, minimum height=10mm] (E1ncoder) {$\gamma^1$}; 
	
	\node[block, right =1cm of k21] [minimum width=8mm, minimum height=10mm] (E2ncoder) {$\gamma^2$}; 
	\node[block, right =1.18cm of k0] [minimum width=2mm, minimum height=25mm, pin={[pinstyle]left:{\small $z(k)$}}] (sum) {MAC}; 
	\node[inner sep=0,minimum size=0, right = 1.3cm of E1ncoder] (e1) {};
	\node[inner sep=0,minimum size=0, right = 1.3cm of E2ncoder] (e2) {};
	\node[block, right =.75cm of sum] [minimum width=6mm, minimum height=25mm] (decoder) {$\delta$}; 
	\node[inner sep=0,minimum size=0, below =  0cm of E1ncoder] (l1) {};  
	\node[inner sep=0,minimum size=0, above =  0cm of E2ncoder] (l2) {};  
	\node[inner sep=0,minimum size=0, right = .9cm of decoder.70 ] (decod0) {}; 
	\node[inner sep=0,minimum size=0, right = .9cm of decoder] (decod1) {};  
	\node[inner sep=0,minimum size=0, right = .9cm of decoder.-70] (decod2) {}; 
	\draw [draw,->] (k12) -- node[below]      {{\small $y^{1}(k)$}} (E1ncoder);
	\draw [draw,->] (k21) -- node             {{\small $y^{2}(k)$}} (E2ncoder);
	\draw [draw,-]  (E1ncoder) -- node[below] {{\small $s^{1}(k)$}} (e1);
	\draw [draw,-]  (E2ncoder) -- node        {{\small $s^{2}(k)$}} (e2);
	\draw [draw,-]  (plant0)   -- node        {{\small $y^{0}(k)$}} (k0);
	\draw [draw,->] (k0)       -- (l1);
	\draw [draw,->] (k0)       -- (l2);
	\draw [draw,->] (e1)       -- (sum);
	\draw [draw,->] (e2)       -- (sum);
	\draw [draw,->] (sum)      -- node {{\small $q(k)$}} (decoder);
	\draw [draw,->] (decoder.70)  -- node[above] {{\small $\hat{x}^{0}(k)$}} (decod0);
	\draw [draw,->] (decoder)  -- node[above] {{\small $\hat{x}^{1}(k)$}} (decod1);
	\draw [draw,->] (decoder.-70) -- node[above] {{\small $\hat{x}^{2}(k)$}} (decod2);
	\end{tikzpicture}}
\caption{State estimation of LTI systems with disturbances over two-input MAC channels.}
\label{fig:setup}
\end{figure}
Two sensors $\gamma^1$ and $\gamma^2$ are deployed to generate the sequences $s^1(k)$ and $s^2(k)$ using the plants' outputs. Although measurement signals $y^1(k)$ and $y^2(k)$ are only seen by their respective encoder, the output of Plant 0 is transmitted to both devices. The aim of such formulation is to capture some of the essential elements of distributed state estimation problem, e.g., each sensor observes a different subset of the overall system's dynamical modes and any possible overlap between these subsets is modelled by the common Plant 0. 

We define the \emph{topological entropy} $h_i$ as 
\begin{align}
h_i := \sum_{ \ell : \vert \lambda^i_{\ell} \vert \geq 1} \log \vert \lambda^i_{\ell} \vert, \ \ \text{for } i = 0,1,2.
 \label{eq:topEntropies}
\end{align}
where $\lambda^i_{\ell}$ are the unstable poles of the $i$-th system. Under suitable and natural assumptions, we show that if there exists a coder-estimator tuple $(\gamma^1 , \gamma^2 , \delta )$ yielding uniformly bounded estimation errors, i.e., between the estimation $\hat{X}(k)$ and the real state $X(k)$, then
\begin{align}
	\boldsymbol{h} \in  \mathcal{C}_{0},
	\label{conj1}
\end{align} 
where $h := \left( h_0 , h_1, h_2 \right)^T $ is the vector of topological entropies (\ref{eq:topEntropies}) of the corresponding systems and $\mathcal{C}_{0}$ denotes the zero-error capacity region of the MAC. On the other hand, if\footnote{The operator $\mathrm{int} \left(\cdot\right)$ denotes the interior of a given region.} $h\in \mathrm{int} \left(\mathcal{C}_0 \right)$, then a coder-estimator tuple that achieves uniformly bounded state estimation errors can be constructed. 
This condition is considered tight as sufficiency and necessity differ only on the boundary, which is a set of Lebesgue measure zero. 

Unlike our previous paper \cite{nairITW2019}, where the two-user MAC model was studied, in this work we derive a channel coding theorem characterizing the zero-error capacity region of an $M$-user MAC with one common message and $M\geq 2$. Additionally, we present here the detailed necessity proof of Theorem \ref{th:StateEstimation}, which appeared as a brief communication in \cite{zafzoufISIT2020}, using the notions of nonstochastic information as well as nonstochastic conditional information, and furthermore prove the sufficient condition, by constructing a suitable scheme. Finally, the setup (\ref{originalState1})-(\ref{originalOutput1}) illustrated in Fig.~\ref{fig:setup} is shown here to be general for noiseless linear systems. 

\subsection{Structure of the Paper}

The rest of the paper is organized as follows. In Section \ref{sec:NonstochasticInformation}, some basic definitions related to the nonprobabilistic framework are introduced and the  MAC model along with the zero-error coding scheme are presented.  
Next, we extend an earlier result on the zero-error capacity region of a two-input, single-output MAC \cite{nairITW2019} to the more general case of $M$ inputs for any given block-length $n$ in Section \ref{sec:ZECapacity}. Along with the results, both converse and achievability proofs are provided in this section. 
Then, the problem of uniform state estimation over two-user MAC is thoroughly investigated in Section \ref{sec:distributedStateEstimation}. Furthermore, an example of state estimation over the \textit{binary adder channel} (BAC) is studied in Section \ref{sec:numExample}.
Finally, Section \ref{sec:conclusion} concludes the article by summarizing the main contributions and discussing possible future directions.

\subsection{Notation}

In this paper, we use upper case letters, e.g., $X$, to denote uncertain variables (uv's), and lower case letters, e.g., $x$, for their realizations. Calligraphic letters, such as $\mathcal{X}$, stand for sets, and the cardinality of set $\mathcal{X}$ is represented by $\vert \mathcal{X} \vert$. Note that a pair of vertical bars, i.e., $\vert a \vert$, is also used to denote the absolute value of a scalar $a$. We represent the sequence $\lbrace x_i \rbrace_{i =m}^n $ by $x_{m:n}$. Furthermore, the max-norm of a vector $v$ is denoted by $\Vert v \Vert$. The expression $[m:n]$ represents the sequence of integers $\lbrace m, m+1, \ldots , n-1, n \rbrace$. Finally, the symbol $\lceil \cdot \rceil $ stands for the ceiling function, and the logarithms throughout this paper are in base 2. 
\section{Preliminaries of Nonstochastic Information}
\label{sec:NonstochasticInformation}

In this section, we briefly review the \emph{uncertain variable} (uv) framework introduced in \cite{nair2013, nair2015}. Using this framework, the problem of zero-error communication over an $M$-user MAC is subsequently analyzed.  
\subsection{Uncertain Variables, Unrelatedness and Markovianity}

Consider the sample space $\Omega$ as shown in Fig.~\ref{fig:notionUV}. A uv $X$ consists of a mapping from $\Omega$ to a set $\mathcal{X}$ \cite{nair2013}. Hence, each sample $\omega\in\Omega$ induces a particular realization $X(\omega)\in \mathcal{X}$. Given a pair of uv's $X$ and $Y$,  the \emph{marginal}, \emph{joint} and \emph{conditional ranges} are denoted as 
\begin{align}
\llb X \rrb & := \lbrace X(\omega): \omega \in \Omega \rbrace \subseteq \mathcal{X}, \\
\llb X,Y \rrb & := \lbrace \left( X(\omega), Y(\omega) \right) : \omega \in \Omega \rbrace \subseteq \mathcal{X} \times \mathcal{Y}, \\
\llb X \vert y \rrb & := \lbrace X(\omega): Y(\omega) = y, \omega \in \Omega \rbrace \subseteq \mathcal{X}.  
\end{align} 

\begin{figure}
\centering
\scalebox{.8}{
	\begin{tikzpicture}
	\draw [fill=gray!20, rotate=0] (-3,0) ellipse (1.3 and .5);
	\draw [thick,-latex] (-1,0) -- (2,0) node [right] {$\mathcal{X}$};
	\draw [-] (0,0.1) -- (0,-0.1);
	\draw [-] (1,0.1) -- (1,-0.1);
	\node at (0.5,-.5) {$\llb X \rrb$};
	\node at (1.3,.3) {$x = X(\omega )$};
	\node at (-2.8,-.8) {\textbf{No measure on $\Omega$ !}};
	\node at (-0.8,.75) {$X$};
	\coordinate (X) at (-2.5,.1);
	\fill (.7,0) circle (2pt);
	\fill (X) circle (2pt) node [below] {$\omega$};
	\node at (-4.7,0) {$\Omega$};
	\path[->] (-2.48,.1) edge [bend left] (.7,0.1);
	\end{tikzpicture}}
\caption{Notion of uncertain variable (uv).}
\label{fig:notionUV}
\end{figure} 
The dependence on $\Omega$ will normally be hidden, with most properties of interest expressed in terms of operations on these ranges. As a convention, uv's are denoted by upper-case letters, while their realizations are indicated in lower-case. The family $\left\{\llb X|y\rrb: y\in\llb Y\rrb\right\}$ of conditional ranges is denoted $\llb X|Y\rrb$.

\begin{definition}[Unrelatedness \cite{nair2013}]
\label{unrelatedDefn}
The uv's $X_1, X_2, \cdots X_n $ are said to be {\em (mutually) unrelated} if
\beq
\llb X_{1}, X_{2}, \cdots, X_n \rrb = \llb X_{1}\rrb \times \llb X_{2}\rrb \times \cdots \times \llb X_{n}\rrb.
\eeq
We denote two unrelated uv's $X_1$ and $X_2$ by $X_1 \perp X_2$.
\end{definition}

\begin{remark}
Unrelatedness is closely related to the notion of qualitative independence \cite{renyi70} between discrete sets. It can be proven that unrelatedness is equivalent to the following conditional range property
\beq
\llb X_{k} | x_{1:k-1} \!] = \llb X_{k}\rrb,  \ \forall x_{1:k-1}\in   \llb X_{1:k-1}\rrb, \  k\in [2:n].
\label{unrelatedconditionalform}
\eeq
\end{remark}

\begin{definition}[Conditional Unrelatedness \cite{nair2013}]
The uv's $X_1,...,X_n$ are said to be {\em conditionally unrelated given} $Y$  if 
\begin{align}
\llb X_1, \ldots , X_n | y\rrb 
= \llb X_1|y\rrb \times\cdots \times \llb X_n| y \rrb , \ 
\forall y\in \llb Y \rrb.
\end{align}
\end{definition}

\begin{definition}[Markovianity \cite{nair2013}]
The uv's $X_1,Y$ and $X_2$ form a {\em Markov uncertainty chain} denoted as $X_1\lra Y \lra X_2$, if 
\beq
\llb X_1|y,x_2\rrb = \llb X_1|y\rrb, \ \ \forall (y,x_2)\in \llb Y,X_2\rrb.
\eeq
The Markov chain $X_1 \leftrightarrow Y \leftrightarrow X_2$ is equivalently denoted by $X_1 \perp X_2 | Y$, i.e., $X_1$ and $X_2$ are unrelated given $Y$.
\label{def:markovianity}
\end{definition}

\begin{remark}
It can be shown that Def.~\ref{def:markovianity} is equivalent to $X_1$ and $X_2$ being conditionally unrelated given $Y$, i.e.,
\beq
\llb X_{1}, X_{2}| y \rrb = 
\llb X_{1} | y \rrb \times \llb X_{2} | y \rrb , \ \forall y\in \llb Y\rrb .
\label{condunrelated}
\eeq
By the symmetry of \eqref{condunrelated}, we conclude that $X_1 \lra Y\lra X_2$  iff $X_2 \lra Y \lra X_1$.
\end{remark}
\subsection{Preliminaries on Nonstochastic Information}

Before presenting the notion of nonstochastic information, we firstly discuss some important background concepts. Throughout this subsection $X$, $Y$, $Z$, $Z'$ and $W$ denote uv's.

\begin{definition}[Overlap Connectedness \cite{nair2013}]
Two points $x \in \llb X\rrb$ and $x' \in \llb X\rrb $ are said to be $\llb X \vert Y \rrb$-overlap connected, denoted $x \leftrightsquigarrow  x'$, if there exists a finite sequence $\lbrace X \vert y_i \rbrace_{i=1}^{m} $ of conditional ranges such that $ x \in \llb X \vert y_1 \rrb,\ x' \in \llb X \vert y_m \rrb$ and $\llb X \vert y_i \rrb \cap \llb X \vert y_{i-1}\rrb \neq \emptyset$, for each $i \in \left[ 2, \cdots, m\right] $. 
\end{definition} 

Obviously, the overlap connectedness is both transitive and symmetric, i.e., it is  an equivalence relation. Thus, it results in  equivalence classes that cover $\llb X\rrb$ and form a unique partition. We call this family of sets the $\llb X |Y \rrb${\em -overlap partition}, denoted by $\llb X \vert Y \rrb_*$.
\begin{figure}
\begin{minipage}{.24\textwidth}
	\scalebox{.8}{
	\begin{tikzpicture}
		\draw [thick,-latex]  (0,0) -- (3.5,0) node [right] {$X$};
		\draw [thick,-latex]  (0,0) -- (0,3)   node [above] {$Y$};  
		
		\fill (1,1) circle (3pt);
		\fill (3,2) circle (3pt);

		\node at (-.5,1) {$y_1$};
		
		\node at (-.5,2) {$y_2$};
		
		\draw [dashed] (1,1) -- (0,1);
		\draw [dashed] (3,2) -- (0,2);
		
		\draw [dashed] (1,1) -- (1,0);
		\node at (1,-.6) {$\llb X \vert y_1 \rrb$};
		
		\draw [dashed] (3,2) -- (3,0);
		\node at (3,-.6) {$\llb X \vert y_2 \rrb$};
		
		\draw [-] (3,-.1) -- (3,0.1);
		\draw [-] (1,-.1) -- (1,0.1);
		\node at (1.8,-1.2) {\textbf{(a)}};
	\end{tikzpicture}}
\end{minipage}
\hfill
\begin{minipage}{.24\textwidth}
	\scalebox{.8}{
	\begin{tikzpicture}
		\draw [thick,-latex]  (0,0) -- (3.5,0) node [right] {$X$};
		\draw [thick,-latex]  (0,0) -- (0,3)   node [above] {$Y$};  
		
		\fill (1,1) circle (3pt);
		\fill (3,2) circle (3pt);
		
		\fill (1,2) circle (3pt);
		
		\node at (-.5,1) {$y_1$};
		
		\node at (-.5,2) {$y_2$};
		
		\draw [dashed] (1,1) -- (0,1);
		\draw [dashed] (3,2) -- (0,2);
		
		\draw [dashed] (1,1) -- (1,0);
		\node at (1,-.6) {$\llb X \vert y_1 \rrb$};
		
		\draw [dashed] (3,2) -- (3,0);
		
		\draw [-] (3,-.1) -- (3,0.1);
		\draw [-] (1,-.1) -- (1,0.1);
		
		\node at (2.5,-.6) {$\llb X \vert y_2 \rrb$};
		\draw [->] (2.5,-.3) -- (3,-.15);
		\draw [->] (2.5,-.3) -- (1,-.15);
		
 		\node at (1.8,-1.2) {\textbf{(b)}};
	\end{tikzpicture}
}
\end{minipage}
\caption{Illustrative example of \textbf{(a)} Overlap disconnected points; \textbf{(b)} Overlap connected points.}
\label{fig:exOC}
\end{figure}
\begin{definition}[Nonstochastic Information \cite{nair2013}]
The {\em nonstochastic information} between $X$ and $Y$ is given by
\beq
I_*[X;Y] = \log_2 \left | \llb X|Y\rrb_*\right |.
\label{defIstar}
\eeq
\end{definition} 

\begin{remark}
Note that the nonstochastic information is symmetric, i.e., $I_*[X;Y] =  I_*[Y;X] $
\end{remark}

\begin{remark}
	Similar to classical information theory, there is an analogous Data Processing Inequality in nonstochastic infromation theory. For any Markov uncertainty chain $W \leftrightarrow X \leftrightarrow Y $, it holds that $I_*[W;Y] \leq   I_*[X;Y]$ \cite{nair2013}. In other words, inner uv pairs in the chain share more information than outer ones.
\end{remark}

\begin{example}
	\emph{Consider uv's $X$ and $Y$ with conditional range family $ \llb X \vert Y \rrb = \lbrace \llb X \vert y_1 \rrb, \llb X \vert y_2 \rrb \rbrace$. Fig.~\ref{fig:exOC}(a) illustrates an example of overlap disconnected points. Observe that $\llb X \vert y_1 \rrb \cap  \llb X \vert y_2 \rrb = \lbrace \emptyset \rbrace $, and hence, the unique overlap partition $\llb X \vert Y \rrb_*$ consists of two singleton sets. Thus, the nonstochastic information in this case is $I_*[X;Y] = \log_2 2 = 1 $ bit.
	On the other hand, Fig.~\ref{fig:exOC}(b) shows the case where the points are connected in overlap sense. In this case, it is easy to see that $\llb X \vert y_1 \rrb \cap  \llb X \vert y_2 \rrb \neq \lbrace \emptyset \rbrace$ and $\llb X \vert Y \rrb_*$ does no longer consist of two singleton sets. Hence, $I_*[X;Y] = \log_2 1 = 0 $ bit.     	    
	}
\end{example} 

\begin{definition}[Common Variables \cite{shannonLatticeTIT1953,wolfITW2004}] \label{defCommonVariable}

A uv $Z$ is said to be a {\em common variable (cv)}  for $X$ and $Y$ if there exist  functions $f$ and $g$ such that
$Z = f(X) = g(Y)$. 

Furthermore, a cv is called {\em maximal} if any other cv $Z'$ admits a function $h$ such that $Z' = h(Z)$.
\end{definition}

\begin{remark}
Note that no cv can take more distinct values than the maximal one.
The concept of a maximal common variable was first presented by Shannon in the framework of random variables \cite{shannonLatticeTIT1953}, to which  he referred by the term ``common information element" for a maximal cv.
\end{remark}

The nonstochastic information $I_*[X;Y]$ is precisely the log-cardinality of the range of a maximal cv between $X$ and $Y$. This is because it can be shown that $\forall (x,y) \in \llb X,Y\rrb$,  the partition set in $\llb X|Y\rrb_*$ that contains $x$ also uniquely specifies the set in $\llb Y|X\rrb_*$ that contains $y$. Thus these overlap partitions define a cv for $X$ and $Y$, with corresponding functions $f$ and $g$ given by the labelling. Furthermore, this cv can be proved to be maximal. See \cite{nair2015} for further details.

\begin{definition}[Conditional $I_*$ \cite{nair2015}] The {\em conditional nonstochastic information} between  $X$ and $Y$ given  $W$  is 
\beq
I_*[X;Y|W]:= \min_{w\in\llb W\rrb } \log_2 \left | \llb X |Y,w \rrb_* \right |,
\label{defcondinfo}
\eeq
where for a given $w\in\llb W\rrb$, $\llb X |Y,w \rrb_*$ is the overlap partition of $\llb X|w\rrb$ induced by the family $\llb X |Y,w \rrb$  of conditional ranges $\llb X |y,w \rrb$, $y\in\llb Y|w\rrb$. 
\end{definition}

\begin{remark}
It can be shown that $I_*[X;Y|W]$ also has an important interpretation in terms of cv's: it is the maximum log-cardinality of the ranges of all cv's $Z = f(X,W) = g(Y,W)$ that are unrelated with $W$.
For more details see \cite{nair2015}.
\end{remark}

\begin{example}
 \emph{To give an intuition to the reader regarding the computation of the conditional nonstochastic information, we show in Fig.~\ref{fig:conditionalRanges} an example of a family $\llb X,Y|W\rrb$ with $W = \{ w_1, w_2\} $. The nonstochastic conditional information in this case is $I_*[X;Y \vert W] := \log_2 (2) = 1 \ \text{bit}.$}
\end{example}

\begin{figure}
	\begin{minipage}{.1\textwidth}
		\begin{flushleft}
			\scalebox{.8}{
			\begin{tikzpicture}
			\draw [thick,-latex]  	(0,0) -- (4.5,0) node [above] {$X$};
			\draw [thick,-latex]  	(0,0) -- (0,3)   node [above] {$Y$};  
			\draw [fill=red!20]   	(.7,2.1)   ellipse (.5 and .4);
			\draw [fill=yellow!20]  (2.3,1.4)      ellipse (.5 and .15);
			\draw [fill=blue!20]  	(3.5,.6)  ellipse (.2 and .5);
			\node at (2.25,3) {\footnotesize $W = w_1$};
			\node at (2.25,-.50) {\footnotesize $\vert \llb X \vert Y,w_1 \rrb_* \vert = 3$};
			\end{tikzpicture}
			}
		\end{flushleft}
	\end{minipage}
	\hspace{30mm}
	\begin{minipage}{.1\textwidth}
		\centering
		\scalebox{.8}{
		\begin{tikzpicture}
		\draw [thick,-latex]  	(0,0) -- (4.5,0) node [above] {$X$};
		\draw [thick,-latex]  	(0,0) -- (0,3)   node [above] {$Y$};  
		\draw [fill=red!20]   	(1.2,.7)   ellipse (.7 and .2);
		\draw [fill=yellow!20]  (2,1.23)      ellipse (.5 and .15);
		\draw [fill=blue!20]  	(3.7,1.8)  ellipse (.5 and .3);
		\node at (2.25,3) {\footnotesize $W = w_2$};
		\node at (2.25,-.50) {\footnotesize $\vert \llb X \vert Y,w_2 \rrb_* \vert = 2$};
		\end{tikzpicture}
		}
	\end{minipage}
	\caption{Illustrative example of a family of conditional ranges for two realizations $w_1$ and $w_2$.}
	\label{fig:conditionalRanges}
\end{figure} 
\section{Error-Free Communication over $M$-User MAC}
\label{sec:ZECapacity}
This section introduces the $M$-user MAC communication system in the nonstochastic framework. Next, the previously discussed concepts are used to obtain an exact characterization of the zero-error capacity region (Theorem \ref{conjecture:DMAC}). This characterization is then used in Section \ref{sec:distributedStateEstimation} to find tight conditions for achieving bounded state estimation errors over a MAC. 
\subsection{System Model}
Consider the communication setup depicted in Fig. \ref{fig:M_User_MAC}. The system consists of $M$ transmitters, each wishing to convey a distinct \emph{private message} $W^j$, where $j \in [1:M]$, and a \emph{common message} $W^0$ to a unique receiver over an $M$-user MAC. Suppose that the messages $W^0, W^1, W^2 \cdots , W^M$ are mutually unrelated and finite-valued. We assume without loss of generality that for $i \in [0:M]$, the messages $W^i$ take the integer values $[1:w_{\max}^i]$ for some integer $w_{\max}^i\geq 1$. For a given block-length $n \geq 1$, the messages are encoded into channel input sequences $X^1_{1:n}, X^2_{1:n}, \cdots , X^M_{1:n}$ as
\beq
X^j_{1:n}= \mathcal{E}^j (W^0,W^j), \ \  j \in [1:M],
\label{code}
\eeq
where $\left\lbrace \mathcal{E}^j \right\rbrace_{j=1}^{M}$ are the coding laws at each transmitter. Note that the \emph{common message} $W^0$ is seen by all encoders while  the {\em private messages} $W^j$ are only available to their respective transmitters.
The code rate for each message is defined as 
\beq
R^i := ( \log_2 w_{\max}^i )/n, \  \ i \in [0:M].
\label{ratedef}
\eeq

This general system configuration, where a common message is seen by all encoders, allows us to incorporate a form of relatedness among the channel input sequences in the model. In the case where the common message can take only one value, so that $R^0=0$, each channel input is generated in isolation and is mutually unrelated with the others. At the other extreme, if the private messages  can each take only one value so that $R^1=R^2=\cdots = R^M=0$, then the channel inputs are generated in complete cooperation. 

The encoded data sequences are then sent through a stationary memoryless MAC as depicted in Fig.~\ref{fig:M_User_MAC}. The output $Y_k\in\mcY$ of the MAC is given in terms of a fixed function $f:\mcX^1\times\cdots\times\mcX^M\times\mcZ\to\mcY$ as
\beq
Y_{k} = f(X^1_{k},X^2_{k}, \cdots ,X^M_{k},Z_k) \in\mcY,  \  \ k\in \mathbb{Z}_{\geq 0},
\label{MAC}
\eeq
where the channel noise is denoted by $Z_k$ and is mutually unrelated with all messages, and past channel noise, i.e., $Z_{1:k-1}$, $W^0$, $W^1, W^2, \cdots , W^M$. We further assume that the range $\llb Z_k\rrb = \mcZ$ is constant. 

The receiver consists of a decoder $\mathcal{D}$ that generates estimates $\hat{W}^0$, $\hat{W}^1, \cdots , \hat{W}^M $ of the transmitted messages using the channel output sequence $Y_{1:n}$. In the context of zero-error communication, these estimates must always be \textit{exactly} equal to the original messages, even with the existence of channel noise or inter-user interference. This requirement means that for any $i \in [0:M]$ the conditional range $\llb W^i|y_{1:n}\rrb$ consists of one element for any channel output sequence $y_{1:n}\in\llb Y_{1:n}\rrb$. This communication system is an extension of the nonstochastic MAC operating with two senders \cite{nairITW2019} to a more general scenario with two or more users. 

For a given code block-length $n$, the {\em operational zero-error $n$-capacity region} $\mcC_{0,n}$ of the MAC is defined as the set of rate tuples $R = (R^i)_{i=0}^M$ for which zero-error communication is possible by suitable choice of coding and decoding functions.  Note that this is well-defined  for finite block-lengths $n$, and so it is of interest in safety-critical low-latency applications. This is unlike the Shannon capacity region $\mathcal{C}$, which requires $n\to\infty$ so as to yield vanishingly small decoding error probabilities.

If we are allowed to use arbitrarily long blocks, i.e., $n \rightarrow \infty $,  then the relevant {\em zero-error capacity region} $\mcC_0$ is given by the closed union
\begin{align} \label{eq:zeCap}
\mcC_0 = \overline{\bigcup_{n \geq 1} \mathcal{C}_{0,n}},
\end{align}
and we define the notion of \textit{achievable rate} as follows. 
\begin{definition}
	A rate tuple $R = (R^0,\ R^1, \cdots,\ R^M)$ is called \emph{achievable} if there exists a sequence of $\left( \left\lceil 2^{n R^0_n}\right\rceil, \ \left\lceil 2^{n R^1_n}\right\rceil, \cdots, \  \left\lceil 2^{n R^M_n}\right\rceil, \ n \right)$ zero-error codes, with $n \in \mathbb{Z}_{\geq 1}$ and $R_n = (R^0_n,\ R^1_n,\ \cdots,\ R^M_n)  \in \mcC_{0,n}$, converging to $R$. 
\end{definition}

We have then the following result.

\begin{theorem}\label{thm:convexity}
	The zero-error capacity region $\mathcal{C}_{0}$ (\ref{eq:zeCap}) is convex. 
\end{theorem} 

\begin{proof}
See Appendix \ref{sec:proofConvexity}.	
\end{proof}
\begin{figure}
\centering
\scalebox{.75}{
\begin{tikzpicture}[auto, node distance=2cm,>=latex'] [scale=1]
\node[block, name=E1] [minimum width=8mm, minimum height=8mm] (E1) {$\mathcal{E}^1$};      
\node[block, below = .2cm of E1] [minimum width=8mm, minimum height=8mm] (E2) {$\mathcal{E}^2$};
\node[block, below of = E2] [minimum width=8mm, minimum height=8mm] (EM) {$\mathcal{E}^M$};
\node at ($(E2)!.5!(EM)$) {\vdots};
\node[block, left = 1.3cm of E1] [minimum width=17mm, minimum height=8mm] (Source1) {Source 1};
\node[block, above =.2cm of Source1] [minimum width=17mm, minimum height=8mm] (Source0) {Source 0};
\node[block, below =.2cm of Source1] [minimum width=17mm, minimum height=8mm] (Source2) {Source 2};
\node[block, below of = Source2] [minimum width=17mm, minimum height=8mm] (SourceM) {Source $M$};
\node at ($(Source2)!.5!(SourceM)$) {\vdots};
\node[inner sep=0,minimum size=0, below = .5cm of E1] (cent) {};  
\node[block, right = 1.5cm of cent] [minimum width=10mm, minimum height=50mm, pin =   
{[pinstyle]above:$Z_{k}$}] (MAC) {MAC};
\node[block, right = 1cm of MAC] [minimum width=10mm, minimum height=50mm] (D) {$\mathcal{D}$};
\node [inner sep=0,minimum size=0, left of=E1] (k1) {};  
\node [inner sep=0,minimum size=0, left of=E2] (k2) {};  
\node [inner sep=0,minimum size=0, right =1cm of Source0](k0) {}; 
\node [circle,fill,inner sep=.7pt, below =.75cm of k0]    (l1) {};  
\node [circle,fill,inner sep=.7pt, below =.9cm of l1]     (l2) {};
\node [inner sep=0,minimum size=0, below =1.9cm of l2]  (l3) {};  
\node [input, left = 1.3cm of E1] (input1) {};
\node [input, left = 1.3cm of E2] (input2) {};
\node [input, left = 1.3cm of EM] (inputM) {};
\node [output, right = 1cm of D.107]  (output1) {};
\node [output, right = 1cm of D.120]  (output2) {};
\node [output, right = 1cm of D]      (output3) {};
\node [output, right = 1cm of D.-104] (outputM) {};
\node at ($(output3)!.4!(outputM)$) {\vdots};
\draw [draw,->] (input1) -- node {$W^1$} (E1);
\draw [draw,->] (input2) -- node {$W^2$} (E2);
\draw [draw,->] (inputM) -- node {$W^M$} (EM);
\draw [draw,-]  (Source0)-- node {$W^0$} (k0);
\draw [draw,-]  (k0)   -- (l3);
\draw [draw,->] (l1)   -- (E1.152);
\draw [draw,->] (l2)   -- (E2.149);
\draw [draw,->] (l3)   -- (EM.144);
\draw [draw,->] (E1)  -- node {$X^1_{k}$} (MAC.119);
\draw [draw,->] (E2)  -- node {$X^2_{k}$} (MAC.190);
\draw [draw,->] (EM)  -- node {$X^M_{k}$} (MAC.-104);
\draw [draw,->] (MAC) -- node {$Y_{k}$}   (D);
\draw [->] ($(D.107)+(1,0)$)  -- node {$\hat{W}^0$}($(output1)+(1,0)$);
\draw [->] ($(D.120)+(1,0)$)  -- node {$\hat{W}^1$}($(output2)+(1,0)$);
\draw [->] (D)                -- node {$\hat{W}^2$}(output3);
\draw [->] ($(D.-104)+(1,0)$) -- node {$\hat{W}^M$}($(outputM)+(1,0)$);
\end{tikzpicture}}
\caption{The $M$-user MAC system with a common message $W^0$ operating at time instant $k$.}
\label{fig:M_User_MAC}
\end{figure}
\subsection{Zero-Error Capacity of $M$-User MAC via Nonstochastic Information}
At this stage, we use nonstochastic information to  establish a multi-letter characterization of the zero-error $n$-capacity region $\mathcal{C}_{0,n}$ for the $M$-user MAC. 
\begin{theorem}\label{conjecture:DMAC}
For a given block-length $n\geq 1$, let  $\mathcal{R}(U,X^1_{1:n}, X^2_{1:n}, \cdots , X^M_{1:n})$ be the set of rate tuples $(R^0,R^1,R^2, \cdots , R^M)$ such that
\begin{align}
n R^0 & \leq  I_{*}[U ;  Y_{1:n}], \label{ineq:1}\\
n R^j & \leq  I_{*}[X^j_{1:n};Y_{1:n} \vert U], \ \forall j \in [1:M], \label{ineq:2}
\end{align}
where $2^{nR^i}$, $i\in[0:M]$, are positive integers,  $X^j_{1:n}$, $j \in [1:M]$, are sequences of inputs to the $M$-user MAC \eqref{MAC}, $Y_{1:n}$ is the corresponding channel output sequence, and $U$ is an auxiliary uv. 

Then, the operational zero-error $n$-capacity region $\mathcal{C}_{0,n}$ of the $M$-user MAC over $n$ channel uses coincides with the union of the regions $\mathcal{R}(U,X^1_{1:n}, X^2_{1:n}, \cdots , X^M_{1:n})$ over all uv's $U,X^1_{1:n}, X^2_{1:n}, \cdots , X^M_{1:n}$ such that 
\begin{itemize}
\item[i)] the sequences $X^j_{1:n}$, $j=1,\ldots , M$, are conditionally unrelated given $U$, 

\item[ii)] $ U \leftrightarrow \lp X^1_{1:n}, X^2_{1:n}, \cdots , X^M_{1:n} \rp  \leftrightarrow Y_{1:n}$ form a Markov uncertainty chain.
\end{itemize}
\end{theorem}   

\begin{proof}
See Section \ref{sec:ProofTheoMAC}.
\end{proof}
	
\begin{remark}
In case of removing the common message $W^0$, $\mcR (X_{1:n}^1, X_{1:n}^2, \cdots , X_{1:n}^M)$ becomes the set of the rate tuples $(R^1,R^2, \cdots , R^M)$ such that 
\begin{align}
n R^j & \leq  I_{*}[X^j_{1:n};Y_{1:n}], \ j \in [1:M], \label{ineq:3}
\end{align}
with $2^{nR^j}$ being positive integers,  $X^j_{1:n}$ referring to the input signals fed into the $M$-user MAC \eqref{MAC} and $Y_{1:n}$ corresponding to the channel output.
The operational zero-error $n$-capacity region $\mathcal{C}_{0,n}$ of the $M$-user MAC over $n$ channel uses in this case coincides with the union of the regions $\mathcal{R}(X^1_{1:n}, X^2_{1:n}, \cdots , X^M_{1:n})$ over all uv's $X^1_{1:n}, X^2_{1:n}, \cdots , X^M_{1:n}$ such that the sequences $X^j_{1:n},\ \forall j \in [1:M]$, are mutually unrelated.

Note that (\ref{ineq:3}) is almost identical to the multi-letter characterization of the {\em ordinary} capacity region of a 2-sender MAC \cite[Thm 4.1]{gamal12}, apart from the use of $I_*$ instead of mutual information.
\end{remark}

\begin{remark} 
In the extreme case of setting $R^1 = \cdots = R^M =0$ bits, whereas the rate $R^0$ is strictly positive, all users $\mcE^1, \cdots , \mcE^M$ transmit the same message, namely the common message $W^0$. Under these circumstances, (\ref{ineq:2}) is trivially satisfied and by applying the \emph{data processing inequality} \cite{nair2013} on the Markov uncertainty chain $ U \leftrightarrow \lp X^1_{1:n}, X^2_{1:n}, \cdots , X^M_{1:n} \rp  \leftrightarrow Y_{1:n}$, it can be seen that the maximum rate $R^0$ is achieved when $U = \left( X^1_{1:n}, \cdots , X^M_{1:n} \right)$. 
\end{remark}
	
\begin{remark} 
In contrast to \cite{nairITW2019}, where the two-transmitter MAC system with a common message was considered, this work provides a coding theorem for the general case of $M$-user MAC. This result is a nonstochastic analogue of Shannon information-theoretic characterization of the ordinary capacity region \cite{slepianWolf1973}. Apart from its multi-letter nature and the use of nonstochastic rather than Shannon information, there are several other notable differences. Firstly, the inequalities are on each individual rate, and not the sum rate. Secondly, the conditioning argument in (\ref{ineq:2}) involves only the auxiliary variable U, and not any of the channel input sequences. Finally, this characterization is valid for finite block-lengths $n$, not just as $n\to\infty$. 

These differences originate from the zero-error requirement on our system, as well as the definition of conditional $I_*$. Note also that the lack of a sum-rate bound does not imply that senders can transmit with rates orthogonal to the others. This is because the zero-error $n$-capacity region is a union of the regions $\mathcal{R}(U, X^1_{1:n},\ldots , X^M_{1:n})$ over the set of all uv's satisfying the conditions in Theorem \ref{conjecture:DMAC}. Though each of these regions is a hypercube aligned with the rate axes, their union may be more complicated than a hypercube.
\end{remark}

\subsection{Proof of Theorem \ref{conjecture:DMAC}}
\label{sec:ProofTheoMAC}

\subsubsection{Converse}

Consider the $M$-user MAC model defined in \eqref{MAC} and let $\left\lbrace R^i \right\rbrace_{i = 0}^{M}$ \eqref{ratedef} be the rates of some zero-error code \eqref{code} with block-length $n$. Furthermore, we set the uv $U=W^0$. By assumption, the messages $\left\lbrace W^i \right\rbrace_{i = 0}^{M}$ are mutually unrelated and hence from \eqref{code} we conclude that the codewords satisfy
\begin{align}
\prod_{j=1}^{M} \llb X_{1:n}^{j} \vert U \rrb = \llb X_{1:n}^{1}, \cdots , X_{1:n}^{M} \vert U \rrb.
\end{align}
Additionally, the unrelatedness of the channel noise $Z$ with the messages $W^i, \ \forall i \in [0:M],$ implies that $Z$ is also unrelated with the codewords $\left\lbrace X^j \right\rbrace_{j = 1}^{M}$. Thus, the Markov chain $Y_{1:n} \lra \left( X^1_{1:n}, X^2_{1:n}, \cdots X^M_{1:n} \right)  \lra U$ is satisfied.

As zero-error communication is assumed, the existence of a decoding function $\mcD^0$ such that 
\begin{align}
W^0 = \mcD^0(Y_{1:n}), 
\end{align}
is then guaranteed. Moreover, since $U=W^0$ it can be directly deduced that $W^0$ is a common variable (cv) (Def.~\ref{defCommonVariable}) between $U$ and $Y_{1:n}$. The maximal cv property of $I_*$ yields the following
\beq
nR^0\equiv \log_2|\llb W^0\rrb| \leq I_*[U; Y_{1:n}].
\eeq
This proves expression \eqref{ineq:1} of Theorem \ref{conjecture:DMAC}.
Next we show inequality (\ref{ineq:2}) for $j \in [1:M]$. Firstly, note that given a specific realization $W^0 = w^0$ of the common message, there exists a unique message $w^1$ associated with the channel codeword $x^1_{1:n}$. This observation follows also from the zero-error property of the chosen code. In fact, if different realizations $W^1 = w^1$ were mapped to the same codeword, then zero-error decoding would obviously be impossible and the assumption would have been violated.
Thus, there certainly exists a function $g^j$ such that
\beq
W^j = g^j(X^j_{1:n}, W^0), \ \text{for} \ j \in [1:M].
\eeq 
Moreover, by the zero-error property there is indeed a decoding function $\mcD^j$ such that 
\beq
W^j = \mcD^j(Y_{1:n}), \ \forall j \in [1:M]. 
\eeq 
Hence, we conclude that the uv $W^j$ is a cv between $(X^j_{1:n}, W^0)$ and $(Y_{1:n}, W^0)$.  
Recall that in the considered MAC model the private messages $W^j$ are unrelated with $U=W^0$ for all $j \in [1:M]$. Therefore, the interpretation of conditional $I_*$ in terms of maximal cv's results in
\begin{align}
nR^j \equiv \log_2|\llb W^j\rrb|   &\leq I_*[X^j_{1:n}; Y_{1:n}|W^0] \nonumber \\
&= I_*[X^j_{1:n}; Y_{1:n}|U], \ \ \forall j \in [1:M], 
\end{align}
proving  \eqref{ineq:2}.

\subsubsection{Achievability}

The achievability proof is established by showing that if we have a set of uv's $U$, $X_{1:n}^j, \ \forall j \in [1:M]$ for some $n \geq 1$ such that the outlined requirements in Theorem \ref{conjecture:DMAC} are fulfilled, then it is possible to construct a zero-error coding scheme at rates achieving equality in \eqref{ineq:1}-\eqref{ineq:2}.

\paragraph{Codebook Generation}

We firstly fix the rate $R^0$ such that $ R^0 \leq (I_*[ U; Y_{1:n}])/n$.
Next, select one point from each set of the family $\llb U \vert  Y_{1:n}\rrb_* $. We then denote the chosen points $u(w^0)$ with $w^0 = \left\lbrace 1, \ldots , 2^{nR^0} \right\rbrace$.

Since $nR^j = I_*[X^j_{1:n};Y_{1:n}|U]$ for $j \in [1:M]$, \eqref{defcondinfo} means that the following inequality holds
\beq
2^{nR^j} \leq \left | \llb X^j_{1:n} | Y_{1:n}, U=u(w^0) \rrb_* \right |, \ \forall j \in [1:M],
\label{RiIneq}
\eeq
with $w^0 \in [1:2^{nR^0}]$. It is therefore possible to select  $2^{nR^j}$ distinct codewords $x^j_{1:n}$ from $\llb X^j_{1:n} | U=u(w^0) \rrb$ for any realization $w^0$ such that each nonempty set of the overlap partition $\llb X^j_{1:n} | Y_{1:n}, U=u(w^0) \rrb_*$ contains exactly one element. Subsequently, these codewords denoted as $\mcE^j(w^0,w^j)$ for $w^j \in [1:2^{nR^j}]$ correspond to the coding laws \eqref{code} where $j \in [1:M]$.

\paragraph{Zero-Error Decoding}

At this stage of the proof, we show that it is possible to achieve zero decoding errors using the presented scheme. 

Firstly, recall that the uv's $U$ and $ \left\lbrace X_{1:n}^{j} \right\rbrace_{j=1}^{M}$ satisfy 
\begin{align}
\prod_{j=1}^{M} \llb X_{1:n}^{j} \vert U \rrb = \llb X_{1:n}^{1}, \cdots , X_{1:n}^{M} \vert U \rrb.
\end{align}
Then, the $n$-tuples of the decoded codewords $\left( \mcE^1 \left(w^0 , w^1 \right), \cdots ,  \mcE^M \left(w^0 , w^M \right)  \right)$ with $w^i \in [1:2^{nR^i}]$ and $\forall i \in [1:M]$ certainly belong to the conditional joint range $\llb X^1_{1:n}, \cdots ,  X^M_{1:n}|  U=u(w^0) \rrb$. This means that any combination of $w^0, w^1, \cdots , w^M$ is mapped to a valid point lying within $\llb X^1_{1:n}, \cdots , X^M_{1:n}, U\rrb$.   
At the receiver, the decoding procedure consists of $M + 1$ stages:
\begin{itemize}
\item[(1)] Firstly, the decoder determines the transmitted common message $w^0$. By construction,  each of the $2^{nR^0}$ points $u(w^0)$ is inside a separate set of the family $\llb U| Y_{1:n}\rrb_*$. Furthermore, recall that the cv property of the overlap partition implies that each set in the family $\llb  U \vert Y_{1:n} \rrb_*$ containing $u$ also uniquely specifies the matching set in $\llb  Y_{1:n}| U\rrb_*$ that contains $y_{1:n}$. Hence, the common message $w^0$ is decoded with zero error. 
\item[(2)] After having found $w^0$, it is now possible to determine which set of the conditional overlap partition $\llb Y_{1:n} | X^1_{1:n}, U=u(w^0) \rrb_*$ contains the sequence $y_{1:n}$. In a similar way as step (1), this set uniquely determines the corresponding set of the family $\llb X^1_{1:n} | Y_{1:n}, U=u(w^0) \rrb_*$ where the codeword $\mcE^1(m^0,m^1)$ lies. Since at most one codeword has been selected from each set of this family for each realization $w^0$, then the private message of user 1, namely $w^1$, is uniquely decoded.  
\item[(3)] In the subsequent $M-1$ stages, the decoder repeats step (2) with $ x_{1:n}^j$ for $j \in \left[ 2 : M \right] $ and similarly recovers $w^j$ with zero error.  
\end{itemize}
Thus, the achievability of Theorem \ref{conjecture:DMAC} is established.
\begin{remark}
	Note that it suffices to take the cardinality of the auxiliary uv $\vert \llb U \rrb \vert \leq \min  \lbrace \vert\llb Y_{1:n}\rrb \vert,\prod_{i=1}^{M}\vert\llb X^i_{1:n}\rrb\vert\rbrace$. 
	To see this, note from the proof above that $U$ can be taken as the common message $W^0$ in a zero-error code for the MAC, which is unambiguously determined by $Y_{1:n}$. Hence the constraint  $|\llb U \rrb | \leq \left | \llb Y_{1:n}\rrb \right |$ may be imposed. Furthermore, $W^0$ is also uniquely determined by the tuple $\left ( X^1_{1:n},\ldots , X^M_{1:n} \right )$  of channel input sequences; otherwise some valid combination of channel input sequences would be associated with multiple $W^0$, violating the zero-error property.  Hence we also need only consider $U$ with $|\llb U \rrb | \leq \prod_{i=1}^M \left| \llb X^i_{1:n}\rrb \right|$. Putting these two constraints together yields the given bound.	
\end{remark}
\begin{remark}
	Note that the achievability proof above is not intended to yield a practical way to find zero-error codes, but rather just prove their existence under the conditions of Theorem \ref{conjecture:DMAC}. 
	We leave it as future work to explore how nonstochastic information ideas could be exploited to construct zero-error codes in practice.
\end{remark}
\section{Distributed State Estimation over Nonstochastic MAC}
\label{sec:distributedStateEstimation}


\subsection{Problem Formulation}

We now consider three discrete LTI dynamical systems characterized by the following system equations for $i \in \lbrace 0,1,2 \rbrace $:
\begin{subequations}
	\begin{align}
	X^i(k+1) &= A^i X^i(k) + V^i(k) \in \mathbb{R}^{d_i}, \label{originalState} \\
	Y^i(k)   &= C^i X^i(k) + W^i(k) \in \mathbb{R}^{b_i}, \label{originalOutput}
	\end{align} 
\end{subequations}
where the uv's $V^i(k)$ and $W^i(k)$ denote process and measurement noise at time instant $k \in \mathbb{Z}_{\geq 0}$. Before being transmitted, the plant output sequences are first encoded into channel input signals  $S^1(k), S^2(k)$ via the coding functions $\gamma^l$ as   
\begin{align}\label{eq:encoding}
	S^l(k) = \gamma^l (k, Y^0(0:k), Y^l(0:k)), \ l \in \lbrace 1,2 \rbrace.
\end{align}
Note that the outputs of System 0 are available to both encoders, whereas Systems 1 and 2 are observed only by their respective users. The encoded data sequences are then sent through a stationary memoryless two-user MAC (\ref{MAC}) as shown in Fig.~\ref{fig:setup}. The received symbol $Q(k)$ is the output of a fixed function $f : \mathcal{S}^1 \times  \mathcal{S}^2 \times \mathcal{Z} \rightarrow \mathcal{Q}$,
\begin{align}\label{eq:MACoutput}
	Q(k) = f(S^1(k), S^2(k), Z(k)) \in \mathcal{Q}, \quad k \in\mathbb{Z}_{\geq 0},
\end{align}
where $Z(k)$ is the channel noise at time $k$. At the receiver side, the symbols are used to generate an estimation $\hat{X}(k) := \begin{pmatrix} \hat{X}^0(k), \hat{X}^1(k), \hat{X}^2(k) \end{pmatrix}^T$ of the original plant states  $X(k) := \begin{pmatrix}  X^0(k) , X^1(k) ,X^2(k) \end{pmatrix}^T$ by means of decoder $\delta$, i.e., 
\begin{align}\label{eq:decoding}
\hat{X}(k) = \delta (k, Q(0:k)) = \begin{pmatrix}  \delta^0 (k, Q(0:k)) \\  \delta^1 (k, Q(0:k)) \\  \delta^2 (k, Q(0:k)) \end{pmatrix} \in \mathbb{R}^d, 
\end{align}
where $d = \sum_{i=0}^2 d_i$ and, for $k=0$, the initial estimate $\hat{X}(0) = 0$. The prediction error is denoted by the uv
\begin{align}
E(k) := \begin{pmatrix} E^0(k) \\ E^1(k) \\ E^2(k) \end{pmatrix} = \begin{pmatrix} X^0(k) -\hat{X}^0(k)\\ X^1(k)-\hat{X}^1(k) \\  X^2(k)-\hat{X}^2(k) \end{pmatrix} \in \mathbb{R}^d.  
\label{eq:estError}
\end{align}
Consider now the following definition.
\begin{definition}[Uniformly Bounded Errors]
For any noise ranges $\llb V(k)\rrb$ and $\llb W(k)\rrb$ with $\sup_{k\geq 0}||\llb V(k)\rrb||<\infty$ and $\sup_{k\geq 0}||\llb W(k)\rrb||<\infty$, $\exists l>0$ such that for any initial condition with range $\llb X(0)\rrb$, that lies within the closed ball $\mathbf{B}_l\subseteq\mathbb{R}^d$ centered at the origin and with radius $l$, it holds 
\begin{align}
\sup_{k\geq 0} \llb \Vert  E(k) \Vert \rrb  =  \sup_{k\geq 0} \llb \Vert  X(k) -  \hat{X}(k) \Vert \rrb < \infty.
\label{eq:uniformBoundedness} 
\end{align}  
\label{def:unifBoundError}	
\end{definition}
\vspace*{-3mm}
The aim is to design the coder-estimator tuple $(\gamma^1, \gamma^2, \delta)$ such that the resulting estimation error is \emph{uniformly bounded} in the sense of Def.~\ref{def:unifBoundError} with respect to the infinity norm. We impose the following assumptions $\forall i,j \text{ and } h \in \lbrace 0,1,2 \rbrace$ and $\forall k,t\in \mathbb{Z}_{\geq 0}$.

\begin{itemize}
\item[\textbf{A1:}] Each matrix pair $(C^i, A^i)$ is observable. 
\item[\textbf{A2:}] The noise terms $V^i(k), W^i(k)$ are uniformly bounded, i.e., $\sup_{k\geq 0}\llb \|V(k)\|\rrb, \sup_{k\geq 0} \llb\|W(k)\|\rrb < \infty$, where $V(k) := \begin{pmatrix} V^0(k), V^1(k) ,V^2(k) \end{pmatrix}^T$ and $W(k) := \begin{pmatrix} W^0(k), W^1(k) ,W^2(k) \end{pmatrix}^T$.
\item[\textbf{A3:}] The initial states $X^i(0)$ and the noise terms $V^j(k), W^h(t)$ are all mutually unrelated, $\forall i,j,h\in \lbrace 0,1,2 \rbrace$ and $\forall k,t\geq 0$. 
\item[\textbf{A4:}] The channel noise signal $Z(k)$ is unrelated with the combined initial state and noise signals $(X(0), V(0:k-1), W(0:k))$.    
\item[\textbf{A5:}] Each system matrix $A^i$ has only strictly unstable eigenvalues, i.e., $\vert \lambda^i_{\ell} \vert > 1$, $\ell = 1, \ldots, d_i$.  
\item[\textbf{A6:}] The zero signal is a valid realization of measurement and process noise, i.e., $0 \in \llb V^i \rrb, \llb W^i \rrb$.
\end{itemize}

\begin{remark}
Assumption \textbf{(A5)} is imposed mainly for the sake of conciseness.  If $A^i$ were allowed to also have strictly stable eigenvalues, then it is straightforward to show that the state components associated with these eigenvalues are uniformly bounded over time. Hence the trivial zero estimator for these state components would yield estimation errors that are uniformly bounded. Thus, without loss of generality, strictly stable eigenvalues and their associated state components may be omitted from each $A^i$.

Eigenvalues with magnitude exactly equal to 1 need more careful treatment. This is because they lead to polynomial rather than exponential growth with time. As such, more delicate techniques than those used in this paper are required. For reasons of space, we therefore also exclude such eigenvalues.
\end{remark}

\begin{remark}
	Notice that although the formulation with three decoupled subsystems may seem special, it will be shown later in Section \ref{subsec:Generality} that this setup is reasonably general for the case of noiseless linear systems. 
\end{remark}
\subsection{Main Result}
We present the main result of this section. First, for each subsystem matrix $A^i\in\RR^{d_i\times d_i}$, let the {\em topological entropy} be
\beq
h_i := \sum_{1\leq \ell \leq d_i: |\lambda_\ell^i| > 1} \log_2|\lambda_\ell^i |. 
\eeq

\begin{theorem} \label{th:StateEstimation}
Consider the linear time-invariant systems in
(\ref{originalState})-(\ref{originalOutput}) whose outputs are coded (\ref{eq:encoding}) and estimated  (\ref{eq:decoding}) via the two-input single output MAC (\ref{eq:MACoutput}). Suppose Assumptions (A1)-(A6) hold. If there exists a coder-estimator tuple $(\gamma^1 , \gamma^2 , \delta )$ yielding uniformly bounded estimation errors (Def. \ref{def:unifBoundError}), then
\begin{align}
h := \begin{pmatrix} h_0, h_1, h_2  
\end{pmatrix}^T \in  \mathcal{C}_{0},
\label{conj}
\end{align} 
where $h$ is the vector of topological entropies of the corresponding systems and $\mathcal{C}_{0}$ refers to the zero-error capacity region (\ref{eq:zeCap}) of the channel.

Furthermore, if $h\in \mathrm{int} \left(\mathcal{C}_0 \right)$, then a coder-estimator tuple that achieves uniformly bounded state estimation errors can be constructed. 	
\end{theorem}

\begin{proof}
See Section \ref{sec:proofs}.
\end{proof}

\begin{remark}
This result is an extension of \cite{matveevSavkin2007, nair2013}, from centralized LTI systems with point-to-point channels, to distributed LTI systems estimated over a MAC.
\end{remark}

\begin{remark}
 The case where $h$ lies exactly on the boundary of the zero-error capacity region, i.e., $h\in\mcC_0 \setminus \mathrm{int} \left(\mathcal{C}_0 \right)$ introduces some technical issues and is not addressed here. 
\end{remark}
\subsection{Proof of Theorem \ref{th:StateEstimation}}
\label{sec:proofs}
\subsubsection{Necessity Proof}
\label{subsec:necessityProof}
Without loss of generality assume that for $i \in \left\lbrace 0,1,2 \right\rbrace $ the state matrix $A^i$ of Plant $i$ is in \emph{real Jordan canonical} form, i.e., it consists of $m$ square blocks on its diagonal such that the $j$-th block is denoted by $A^i_j \in \mathbb{R}^{d_j \times d_j}$ with $j \in [1:m]$:
\begin{align}
A^i = \begin{pmatrix}
A_1^i & \boldsymbol{0}   & \cdots & \boldsymbol{0} \\
\boldsymbol{0}   & A_2^i & \cdots & \boldsymbol{0} \\
\vdots           & \vdots           & \ddots & \vdots \\
\boldsymbol{0}   & \boldsymbol{0}   & \cdots & A_m^i \\
\end{pmatrix} \in \mathbb{R}^{d_i \times d_i}.
\end{align}
In the following analysis, we consider only the unstable eigenvalues $\lbrace \lambda^i_\ell \rbrace_{\ell=1}^{d_i}$. 
In Appendix \ref{sec:proofIneq}, we establish the following inequality for $k\rightarrow\infty$ and $i \in  \lbrace 0,1,2 \rbrace$:
\begin{align}
\frac{I_*[X^i(0:k -1); Q(0:k - 1)]}{k} &\geq  \frac{I_*[X^i(0); Q(0:k - 1)]}{k} \nonumber \\
&\geq \sum_{\ell=0}^{d_i} \log \vert \lambda^i_{\ell} \vert. 
\label{eq:NecessityIneq} 
\end{align}
Before proceeding with the proof, we present the following lemma. 
\begin{lemma} \label{lemma:conditionalIstar}
Let $\Lambda, \Omega$ and $\Theta$ denote three uv's such that $ \Lambda$ and $\Theta$ are mutually unrelated. Then, the following relationship between $I_*[\Lambda ; \Omega \vert \Theta]$ (\ref{defIstar}) and $I_*[\Lambda ; \Omega]$ (\ref{defcondinfo}) holds 
	\begin{align}
	I_*[\Lambda ; \Omega \vert \Theta] \geq I_*[\Lambda ; \Omega]. 
	\end{align}
\end{lemma}
\begin{proof}
	See Appendix \ref{sec:appLemma}.
\end{proof}    
Since $X^0(0:k-1)$ and $X^l(0)$ are unrelated for $l \in \lbrace 1,2 \rbrace$, it follows from Lemma \ref{lemma:conditionalIstar}
\begin{align}
I_*[X^l(0);Q(0:k-1) \vert X^0(0:k-1)] \geq I_*[X^l(0);Q(0:k-1)].
\label{ineq:proofStep2}  
\end{align}
Note that $S^1(0:k-1) \leftrightarrow X^0(0:k-1) \leftrightarrow S^2(0:k-1)$, i.e., $S^1(0:k-1) \perp S^2(0:k-1) \vert X^0(0:k-1)$. Furthermore, the initial states $\left\lbrace X^i(0)\right\rbrace_{i=0}^{2}$ and additive noises $\left\lbrace W^i(k), V^i(k) \right\rbrace_{i=0}^{2} $ are mutually unrelated. Hence, the requirement \textbf{(A3)} results in the Markov chain $X^l(0) \leftrightarrow S^l(0:k - 1) \leftrightarrow Q(0:k - 1 ) \vert  X^0(0:k-1)$, for $l \in \lbrace 1,2 \rbrace $. Thus, the \textit{conditional data processing inequality} \cite{nair2015} yields 
\begin{align}\label{dataProcessingIneq}
I_*[X^l(0);Q&(0:k - 1) \vert X^0(0:k-1)]  \leq 	 \nonumber \\
&I_*[S^l(0:k - 1); Q(0:k - 1)\vert X^0(0:k-1)].                       
\end{align}  
By combining this lower bound on $I_*[S^l(0:k - 1); Q(0:k - 1)\vert X^0(0:k-1)]$ with inequalities (\ref{eq:NecessityIneq}) and (\ref{ineq:proofStep2}), we obtain 
\begin{align}
\frac{I_*[S^l(0:k -1); Q(0:k - 1) \vert X^0(0:k -1)]}{k} &> \sum_{\ell=0}^{d_l} \log \vert \lambda^l_{\ell} \vert,
\label{eq:lastIneqNecessity}
\end{align} 
for $k \rightarrow \infty $. From (\ref{eq:NecessityIneq}) and (\ref{eq:lastIneqNecessity}), we conclude that
\begin{align}
\frac{I_*[X^0(0:k -1); Q(0:k - 1)]}{k} > \sum_{\ell=0}^{d_0} \log \vert \lambda^0_{\ell} \vert, \nonumber \\
\frac{I_*[S^l(0:k -1); Q(0:k - 1) \vert X^0(0:k -1)]}{k} > \sum_{\ell=0}^{d_l} \log \vert \lambda^l_{\ell} \vert,\nonumber
\end{align}
for $l \in \left\lbrace  1,2 \right\rbrace$. This completes the proof of necessity. 
\subsubsection{Sufficiency Proof}

\begin{figure}[t]
\centering
\vspace*{.33cm}
\scalebox{0.75}{
	\begin{tikzpicture}[auto, node distance=2cm,>=latex'] [scale=2]
	\node[block, name=O1] [minimum width=8mm, minimum height=8mm] (O1) {$\mathcal{O}^l$};      
	\node[block, below = .5cm of O1] [minimum width=8mm, minimum height=8mm] (O0) {$\mathcal{O}^0$};

	\node[block, right = 1cm of O1] [minimum width=8mm, minimum height=8mm] (down1) {$\downarrow n$};
	\node[block, right = 1cm of O0] [minimum width=8mm, minimum height=8mm] (down0) {$\downarrow n$};
	
	\node[block, right = 1cm of down1] [minimum width=8mm, minimum height=8mm] (Q1) {$\mathcal{Q}^l$};
	\node[block, right = 1cm of down0] [minimum width=8mm, minimum height=8mm] (Q0) {$\mathcal{Q}^0$};
	
	\node (midway) at ($(Q1)!.5!(Q0)$) {};
	
	\node[block, right = 1cm of midway] [minimum width=8mm, minimum height=20mm] (Enc) {$\mathcal{E}^l$};
	\node[inner sep=0,minimum size=0, left = .5cm of O1] (input1) {};  
	\node[inner sep=0,minimum size=0, left = .5cm of O0] (input0) {};  
	\node[output, right = .5cm of Enc]  (output) {};
	
	\node[inner sep=0,minimum size=0, left = 0cm of input1]() {$Y^l(k)$};
	\node[inner sep=0,minimum size=0, left = 0cm of input0]() {$Y^0(k)$};
	
	\node[inner sep=0,minimum size=0, right = 0cm of output]() {$S^l(nk:n(k+1)-1)$};
	\draw [draw,->] (input1) -- (O1);
	\draw [draw,->] (input0) -- (O0);
	
	\draw [draw,->] (O1) -- node {$\bar{X}^l(k)$} (down1);
	\draw [draw,->] (O0) -- node {$\bar{X}^0(k)$} (down0);
	
	\draw [draw,->] (down1) -- node {$\bar{X}^l(nk)$} (Q1);
	\draw [draw,->] (down0) -- node {$\bar{X}^0(nk)$} (Q0);
	
	\draw [draw,->] (Q1) -- node {$M^l$} (Enc.122);
	\draw [draw,->] (Q0) -- node {$M^0$} (Enc.-122);
	
	\draw [draw,->] (Enc) -- (output);
	
	\node (midway0) at ($(O1)!.5!(O0)$) {};
	\node[inner sep=0,minimum size=0, left = .6cm of midway0](midway01){};
	
	\node[inner sep=0,minimum size=0, above = 1.5cm of midway01](a){};
	\node[inner sep=0,minimum size=0, below = 3cm of a](b){};
	\node[inner sep=0,minimum size=0, right = 6.5cm of b](c){};
	\node[inner sep=0,minimum size=0, above = 3cm of c](d){};
	
	\node (midwayText) at ($(a)!.5!(d)$) {};
	\node[inner sep=0,minimum size=0, above = 0cm of midwayText](){Encoder $ \gamma^l$};
	
	\draw [draw,dashed] (a) -- (b) -- (c) -- (d) -- (a);
	
	\end{tikzpicture}
}
\caption{Structure of encoder $\gamma^l$ for $l \in \lbrace 1, 2 \rbrace$.}
\label{fig:EncoderStructure}
\end{figure}
The sufficiency of (\ref{conj}) is now established using a separation structure between source and channel coding. To this end, we discuss in detail the structure of the encoder blocks $\gamma^1$ and $\gamma^2$ as depicted in Fig.~\ref{fig:EncoderStructure}. 

Firstly, a Luenberger observer $\mcO^i$, with $i \in \lbrace 0,1,2 \rbrace$ generates the signals $\bar{X}^i(k)$ at time instant $k$. The state observer $\mcO^i$ is defined by the following equation 
\begin{align}
\bar{X}^i(k+1) = A^i \bar{X}^i(k) + L^i \left( Y^i(k) - C^i\bar{X}^i(k)\right),
\label{eq:observer}
\end{align} 
where $L^i$ is a filter matrix of appropriate dimensions. 
The observability of $(A^i, C^i)$ for $i \in \lbrace 0,1,2 \rbrace$ -Assumption \textbf{(A1)}- implies that there exists an observer (\ref{eq:observer}) which guarantees asymptotic boundedness of estimation error. More details regarding linear state observers can be found in \cite{linSys}. Note that the process noise corresponds to the innovations fed to the state observer, i.e., 
\begin{align}
	\bar{V}^i(k) = L^i (Y^i(k) - C^i\bar{X}^i(k)).
\end{align}
The $i$-th observer's output $\bar{X}^i(k+1)$ are then down-sampled by a factor of $n$ to yield 
\begin{align}
\bar{X}^i(n(k+1)) = (A^i)^n \bar{X}^i(nk) + \Psi_{n-1}^i(k), 
\end{align}
where it can be shown that the accumulated innovation noise $\Psi_{r}^i(k) := \sum_{\xi = 0}^{r} (A^i)^{n-1-\xi} \bar{V}^i(nk + \xi)$ is uniformly bounded $\forall r \in [0:n-1]$ over $k \in \mathbb{Z}_{\geq 0}$ when an appropriate filter $L^i$ is used.

Next, each of the down-sampled sequences is processed by the respective adaptive quantizer $\mcQ^i$ to generate the messages $M^i$ that are drawn from finite sets $\mcM^i$ with cardinalities $\vert \mcM^i \vert$ for $i \in \lbrace 0,1,2 \rbrace$. Finally, each pair of messages $\left\lbrace\left(M^0,M^l\right)\right\rbrace_{l=1}^2$ is mapped into the codeword $S^l(nk:n(k+1)-1)$ with the block-length $n$ by the corresponding channel encoder $\mcE^l$ at time instant $k$. The code rate for each message $M^i$ is then
\begin{align}
R^i = (\log \vert \mcM^i \vert)/n, \ \text{for } i\in\lbrace 0,1,2\rbrace. 
\end{align}

The receiver $\delta$ is a three-stage process that consists of the reverse operations performed by the encoding block. Firstly, using the channel output $Q(nk:n(k +1)-1)$, the message estimates $\lbrace \hat{M}^i \rbrace_{i=0}^{2}$ are produced by an appropriate channel decoder. Each of these estimates is then processed by an \textit{adaptive dequantizer} $\mcG^i$ followed by an upsampling operation by $n$. This interpolation process consists in inserting $n-1$ zeros between each couple of samples in $\hat{X}^i(nk)$ to generate the state estimations $\hat{X}^i(k)$ at time instant $k$ for $i\in \lbrace 0,1,2 \rbrace$. 

\begin{figure}[t!]
	\centering
	\scalebox{0.85}{
		\begin{tikzpicture}[auto, node distance=2cm,>=latex'] [scale=2]
		\node[block, name=G1] [minimum width=8mm, minimum height=8mm] (G1) {$\mathcal{G}^1$};      
		\node[block, below = .5cm of G1] [minimum width=8mm, minimum height=8mm] (G0) {$\mathcal{G}^0$};
		\node[block, below = .5cm of G0] [minimum width=8mm, minimum height=8mm] (G2) {$\mathcal{G}^2$};

		\node[block, left = 1cm of G0] [minimum width=8mm, minimum height=34mm] (D) {$\mathcal{D}$};
		
		\node[block, right = 1.3cm of G1] [minimum width=8mm, minimum height=8mm] (up1) {$\uparrow n$};
		\node[block, right = 1.3cm of G0] [minimum width=8mm, minimum height=8mm] (up0) {$\uparrow n$};
		\node[block, right = 1.3cm of G2] [minimum width=8mm, minimum height=8mm] (up2) {$\uparrow n$};
		
		\node[inner sep=0,minimum size=0, left = .5cm of D] (input) {};   
		
		\node[output, right = .5cm of up1]  (output1) {};
		\node[output, right = .5cm of up0]  (output0) {};
		\node[output, right = .5cm of up2]  (output2) {};
		
		\node[inner sep=0,minimum size=0, left = 0cm of input]() {$Q(nk:n(k+1)-1)$};
		
		\node[inner sep=0,minimum size=0, right = 0cm of output1]() {$\hat{X}^1(k)$};
		\node[inner sep=0,minimum size=0, right = 0cm of output0]() {$\hat{X}^0(k)$};
		\node[inner sep=0,minimum size=0, right = 0cm of output2]() {$\hat{X}^2(k)$};
		\draw [draw,->] (input) -- (D);
		
		\draw [draw,->] (D.73) -- node {$\hat{M}^1$} (G1);
		\draw [draw,->] (D) -- node {$\hat{M}^0$} (G0);
		\draw [draw,->] (D.-73) -- node {$\hat{M}^2$} (G2);
		
		\draw [draw,->] (G1) -- node {$\hat{X}^1(nk)$} (up1);
		\draw [draw,->] (G0) -- node {$\hat{X}^0(nk)$} (up0);
		\draw [draw,->] (G2) -- node {$\hat{X}^2(nk)$} (up2);
		
		\draw [draw,->] (up1) -- (output1);
		\draw [draw,->] (up0) -- (output0);
		\draw [draw,->] (up2) -- (output2);
		
		\node[inner sep=0,minimum size=0, right = .2cm of input](midway0){};
		
		\node[inner sep=0,minimum size=0, above = 2cm of midway0](a){};
		\node[inner sep=0,minimum size=0, below = 4cm of a](b){};
		\node[inner sep=0,minimum size=0, right = 5.2cm of b](c){};
		\node[inner sep=0,minimum size=0, above = 4cm of c](d){};
		
		\node (midwayText) at ($(a)!.5!(d)$) {};
		\node[inner sep=0,minimum size=0, above = 0cm of midwayText](){Decoder $\delta$};
		
		\draw [draw,dashed] (a) -- (b) -- (c) -- (d) -- (a);
		
		\end{tikzpicture}
	}
	\caption{Structure of decoder $\delta$.}
	\label{fig:DecoderStructure}
\end{figure}

As $\text{int}(\mcC_0)$ is, by definition, an open and non-empty set, there exists an open $l_\infty$-ball centered at $h$ and with arbitrarily small radius $\delta >0$ denoted as $\textbf{B}_\delta (h) \subseteq \text{int}(\mcC_0)$. Hence, for the given vector of topological entropies $ h = (h_0,h_1,h_2)^T \in \text{int}(\mcC_0)$, there exists a rate vector $R = (R^0,R^1,R^2)^T \in \textbf{B}_\delta (h)$, i.e., for $i \in  \lbrace 0,1,2 \rbrace$, $R^i = h_i + \delta$.
Let $\textbf{B}_\epsilon (R)$ be the $l_\infty$-ball of center $R$ and arbitrarily small radius $\epsilon >0$ as depicted in Fig.~\ref{fig:sufficiencyProof}. Since $R \in \mcC_0 $, then there exists a rate vector $R_n \in \mcC_{0,n} \cap \textbf{B}_\epsilon (R)$ with $n \in \mathbb{Z}_{\geq 1} $, i.e., 
\begin{align}
	\Vert R - R_n \Vert_{\infty} \leq \epsilon.
\end{align}
As $\epsilon$ can be arbitrarily chosen, we can select $\epsilon$ small enough such that $\epsilon < \delta$; and hence, it is guaranteed that there exists a zero-error code with rate vector $R_n$ that is component-wise strictly larger than $h$.

Consequently, the communication channel linking each message $M^i$ to its estimate $\hat{M}^i$ can be modeled as noiseless. Hence, by the data rate theorem (Proposition 5.2 in \cite{tatikonda2004}) $\forall R_n^i > h_i$, $\exists \mcQ^i, \mcG^i$ such that the prediction  error $E^i(kn) = X^i(kn) - \hat{X}^i(kn)$ is uniformly bounded. 
\begin{figure}
\centering
	\scalebox{0.75}{
	\begin{tikzpicture}[scale=1]
	\begin{axis}[%
		xlabel={$R^1$},
		ylabel={$R^2$},
		xmin = 0, 
		xmax = 1,
		ymax = 1, 
		ymin = 0,
		xticklabels={,,},
		yticklabels={,,}
		]
		\addplot[name path=f,domain=0:1,blue] {-x+1};
		
		\path[name path=axis] (axis cs:0,0) -- (axis cs:1,0);
		
		\addplot [
		thick,
		color=blue,
		fill=blue, 
		fill opacity=0.05
		]
		fill between[
		of=f and axis,
		split,
		every segment no 0/.style={
			blue,
		},
		];
 		\node at (axis cs:  .2,  .2) {$\mcC_0$};
	\end{axis}
	\node at (3,3) [below right] {$R$};
	\fill[black] (3,3) circle (1.5pt);
	\node at (2,2) [below] {$h$};
	\fill[black] (2,2) circle (1.5pt);
	
	\fill[black] (2.5,2.7) circle (1.5pt);
	
	\draw[draw,dashed,red,thick] (3,3) +(-.5,-.5) rectangle +(.5,.5) ;
	
	\draw [draw,dashed] (2,2) -- (3,2) -- (3,3) -- (2,3) -- (2,2);
	\draw[latex-latex] (2,1.9) -- (3,1.9);
	\node at (2.5,1.9) [below] {\small $\delta$};
	
	\draw[latex-latex, red] (2.5,3.6) -- (3.5,3.6);
	\node at (3,3.6) [above] {\textcolor{red}{\small $2\epsilon$}};
	
	\draw[-latex] (2,3.2) -- (2.5,2.7);
	\draw[-] (2,3.2) -- (1.8,3.2);
	
	\node at (1.8,3.2) [left] {$R_n \in \mcC_{0,n}$};
	
	\draw[-latex,red] (4,3.2) -- (3.5,3.2);
	\node at (4,3.35) [right] {\textcolor{red}{\small $\Vert \cdot \Vert_{\infty}$-ball of radius}};
	\node at (4,3.05) [right] {\textcolor{red}{\small $\epsilon$ and center $R$}};
	
	\draw[-latex,blue] (2,4.5) -- (1.5,4.5);
	\node at (2,4.5) [right] {{\small Boundary of the $\mcC_0$ region}};
	\end{tikzpicture}}
\caption{Illustrative figure of the zero-error capacity region's boundary in a two-dimensional space, i.e., $R^0 = 0$. The region shadowed in blue corresponds to a part of the channel's [zoomed-in] zero-error capacity region $\mcC_0$.}
\label{fig:sufficiencyProof}
\end{figure}
Now, every time instant $t \in \mathbb{Z}_{\geq 1}$ can be written for some nonnegative integer $k$ as $t = kn + r$, where $r \in  [0 : n - 1]$. Consider also the estimator
\begin{align}
\hat{X}^i(t) := (A^i)^r \hat{X}^i(nk), \ \text{for } i \in \lbrace 0,1,2\rbrace. 
\end{align}
We examine the requirement (\ref{eq:uniformBoundedness}) for the resulting estimation error $E^i(t) = X^i(t)- \hat{X}^i(t)$, i.e.,
\begin{align}
\sup & \Vert E^i(t)  \Vert = \sup \Vert (A^i)^r X^i(nk) + \Psi_{r}^i(k) - (A^i)^r \hat{X}^i(nk) \Vert \nonumber \\ 
& \leq \Vert (A^i)^r \Vert  \sup \Vert X^i(nk)- \hat{X}^i(nk) \Vert + \Vert \Psi_{r}^i(k) \Vert \nonumber \\
& \leq \max_{r\in[0:n-1]} \left\lbrace \Vert (A^i)^r \Vert \right\rbrace \sup \Vert E^i(nk) \Vert + \Vert   \Psi_{r}^i(k) \Vert, 
\label{proof:suffieciency4}
\end{align}
where the prediction error $E^i(nk)$ of the down-sampled system was shown to be uniformly bounded for some coder-estimator tuple, and the accumulated process noise term $\Psi_{r}^i(k)$ does also satisfy this condition. The RHS of (\ref{proof:suffieciency4}) is therefore uniformly bounded over $k \in \mathbb{Z}_{\geq 0}$. This completes the sufficiency proof.

\subsection{Generality of the Proposed Setup for Noiseless Linear Systems}
\label{subsec:Generality}
\begin{figure}
	\centering
	\scalebox{0.7}{
	\begin{tikzpicture}[auto, node distance=2cm] 
	\node[block, name= Enc1] [minimum width=20mm, minimum height=8mm] (Enc1) {Sensor $\gamma^1$};
	
	\node[block, below = .7cm of Enc1] [minimum width=20mm, minimum height=8mm] (EncM) {Sensor $\gamma^2$};
	
	\node (midpoint) at ($(Enc1)!0.5!(EncM)$) {};
	
	\node[cloud node, left = 2cm of midpoint] [minimum width=25mm, minimum height=25mm, align=center] (Sys) { };
	
	\node[inner sep=0,minimum size=0, below = .5cm of Sys, align=center] (system) {Noiseless LTI System \\ with \textbf{two} mode subsets}; 
	
	\node (circ1) at ($(Sys)-(.1,.2)$) { };
	\node (circ2) at ($(circ1)+(.4,.7)$) { };
	
	\begin{scope}[transparency group]
	\begin{scope}[blend mode=multiply]
	
	\draw[fill=blue!10] (circ1) circle (.5cm);
	\draw[fill=red!10]  (circ2) circle (.5cm);
	
	\node (circ1) at ($(Sys)-(.1,.2)$) {$\Lambda^2$};
	\node (circ2) at ($(circ1)+(.4,.7)$) {$\Lambda^1$};
	
	\end{scope}
	\end{scope}
	\node[block, right = 1.5cm of midpoint, align=center] [minimum width=10mm, minimum height=20mm, pin={[pinstyle]above:{\small  Bounded Noise}}] (MAC) {\rotatebox{90}{Multiple Access Channel}};
	
	\node[block, right = .5cm of MAC] [minimum width=8mm, minimum height=39mm] (D) {\rotatebox{90}{Decoder}};
	\node[inner sep=0,minimum size=0, left = .7cm of Enc1] (in1) {};
	
	\node[inner sep=0,minimum size=0, left = .7cm of EncM] (inM) {};
	
	\node[inner sep=0,minimum size=0, right = .6cm of Enc1] (out1) {};
	
	\node[inner sep=0,minimum size=0, right = .6cm of EncM] (outM) {};
	\node [output, right = .7cm of D.-60] (output1) {};
	\node [output, right = .7cm of D] (output2) {};
	\node [output, right = .7cm of D.60] (outputM) {};
	\draw [arrow, align=center] (in1) -- node {$y_1$} (Enc1);
	
	\draw [arrow] (inM) -- node {$y_2$} (EncM);
	\draw [-] (Sys) -- node { } (in1);
	
	\draw [-] (Sys) -- node { } (inM);
	\draw [arrow] (Enc1) -- node { } (out1);
	\draw [arrow] (EncM) -- node { } (outM);
	\draw [arrow] (MAC)  -- node {} (D);
	\draw [arrow] (D) -- node {$\hat{x}$} (output2);
	\end{tikzpicture}}
	\caption{The states of a noiseless (process and measurement) LTI system are detected by two distinct sensors. Each sensor incorporates an observer and a channel encoder. In this scenario, we consider uncoordinated access strategies with no apriori resource allocation modeled as a MAC.}
	\label{fig:genModel}
\end{figure}

The dynamic decoupling between the three subsystems in (\ref{originalState})-(\ref{originalOutput}) may seem like a strict requirement. In this subsection, we show that any noiseless linear plant  observed via two sensors can be put into this form, under a mild assumption on the output matrices. Subsystems 1 and 2 represent the plant modes that are observable only at sensors 1 and 2 respectively, while subsystem 0 represents the modes that are observable at each of the two sensors. 

Consider a noiseless LTI system with a dynamic matrix $\tilde{A}\in\mathbb{R}^{d\times d}$ that is observed by two different sensors $\gamma^1$ and $\gamma^2$ as depicted in Fig.~\ref{fig:genModel}. The system is described by
\begin{subequations} 
	\begin{align}
	\tilde{X}(k+1) &= \tilde{A} \tilde{X}(k),  \label{eq:stateEq}\\
	Y^l(k) &= \tilde{C}^l \tilde{X}(k), \ \text{for}\ l\in \lbrace 1,2\rbrace, 	\label{eq:outEq}
	\end{align} 
\end{subequations}
where the joint system $\left(\left(\tilde{C}^1, \tilde{C}^2 \right), \tilde{A}\right)$ is observable. In order to decouple the states of the observed system, it is possible to put the matrix $\tilde{A}$ into \textit{Jordan canonical form} \cite{roger2012}. Note that we call \textit{system modes} to refer to the states associated with a Jordan block. Supposing that $\tilde{A}$ has $p$ real eigenvalues $\lbrace \lambda_k\rbrace_{k=1}^{p}$ and $(n-p)/2$ conjugate complex eigenvalues $\lbrace \lambda_k = \alpha_k + j\beta_k \  (\lambda_k^{*} = \alpha_k - j\beta_k)\rbrace_{k=p+1}^{m}$ with $m=p+(n-p)/2$, then there exists a transformation matrix $T$ such that
\begin{align}
T = \left(v_1, \cdots, v_p, \Re(v_{p+1}),\Im(v_{p+1}), \cdots, \Re(v_{m}),\Im(v_{m})\right),
\label{eq:transMatrix}
\end{align}
where $v_k$ refer to the eigenvector corresponding to the eigenvalue $\lambda_k$. Using the matrix $T$, (\ref{eq:stateEq})-(\ref{eq:outEq}) become 
\begin{subequations}
	\begin{align}
		X(k+1) 	&= \underbrace{T^{-1}\tilde{A}T}_{A} X(k),  \label{eq:stateEqTrans}\\
		Y^l(k) &= \underbrace{\tilde{C}^lT}_{C^l} X(k), \ \text{for}\ l\in \lbrace 1,2\rbrace \label{eq:outEqTrans}
	\end{align}
\end{subequations}
with the initial state $X(0) = T^{-1} \tilde{X}(0)$. The dynamic matrix $A$ has then the following form 
\begin{align}
A 	= T^{-1}\tilde{A}T = 
\begin{pmatrix}
J_1 & & 0 \\
& \ddots &  \\
0 & & J_\Gamma
\end{pmatrix},
\label{eq:AmatrixJordan}
\end{align}
with $\Gamma$ \textit{Jordan blocks} such that the $i$-th block is
\begin{align}
J_i = \begin{pmatrix}
\lambda_i 	&1 			& 		&0 	\\
			&\ddots 	&\ddots &  		\\
			&   		&\lambda_i &1	\\
0 			& 			&   	&\lambda_i
\end{pmatrix},
\end{align}
for real eigenvalues, and  
\begin{align}
J_i = \begin{pmatrix}
W_i & I_2 & & 0 \\
&\ddots &\ddots &  \\
&  &  W_i&  I_2 \\
0 	& 		&   &W_i	
\end{pmatrix}, \  W_i = \begin{pmatrix} \alpha_i & \beta_i \\ -\beta_i & \alpha_i \end{pmatrix}, \ I_2 = \begin{pmatrix} 1 & 0 \\ 0 & 1 \end{pmatrix}  \nonumber 
\end{align}
for complex conjugate eigenvalues $\alpha_i\pm j\beta_i$. For simplicity we suppose that each of the blocks $\lbrace J_i \rbrace_{i=1}^\Gamma$ admits distinct eigenvalues. Note that the output matrix $C^l$ can be expressed as a concatenation of sub-matrices $C^l_i$  with $i\in [1:\Gamma]$
\begin{align}
C^l = \begin{pmatrix}
C^l_1	&C^l_2 &\cdots	&C^l_\Gamma		
\end{pmatrix} \in \mathbb{R}^{b_l\times d},\label{eq:Cmatrix}
\end{align}
where each sub-matrix $C^l_i$ is a collection of column vectors such that 
\begin{align}
C^l_i = \begin{pmatrix}
c^l_{1,i}	&c^l_{2,i}	&\cdots	&c^l_{d_{i},i}	
\end{pmatrix}\in \mathbb{R}^{b_l\times d_{i}}. \label{eq:CBlock}
\end{align}  
 
\begin{lemma}\label{lemma:ObseMatrixAC}
	Without loss of generality, consider the state matrix $A$ to be in Jordan form (\ref{eq:AmatrixJordan}) and the output matrix $C^l$ (\ref{eq:Cmatrix}) for $l\in\lbrace 1,2\rbrace$. $\Gamma$ denotes the number of Jordan blocks such that each block admits a distinct eigenvalue. Then the pair $(A,C^l)$ is observable if and only if the leading column of the block $C^l_{i}$, namely $c^l_{1,i}$, is non-zero $\forall i\in [1:\Gamma]$. 
\end{lemma}

\begin{proof}
	See Appendix D.
\end{proof}

We now assume that the matrix $C^l$ has the following structure: If there exists a non-zero coefficient $c^l_{\ell,i}$, $\ell \in [1:d_i]$, associated with a Jordan block $J_i$, then the leading column in $C^l_i$ must be non-zero, i.e., $c^l_{1,i} \neq 0$. 
For each of the systems $\lbrace(A,C^l)\rbrace_{l=1}^2$, constrain the matrix $A$ to the Jordan blocks that have non-zero contribution to the output $y^l$. This results w.l.o.g. in a reduced system $(A^l_*,C^l_*)$ for $l\in\lbrace 1,2\rbrace$. Based on the structure of $C^l$ (and subsequently $C^l_*$) as well as Lemma \ref{lemma:ObseMatrixAC}, both systems are observable and thus the modes associated with the Jordan blocks are reconstructable. To clarify this structure for the reader, an example is provided. 
\begin{example}
	In this example, we consider a system matrix $A$ with two Jordan blocks as follows
	\begin{align}
	A &= 
	\begin{pmatrix}
	\lambda_1 	&1 			&0				&0	\\
	0			&\lambda_1 	&0 				&0  \\
	0		   	&0			&\lambda_2		&1	\\
	0		   	&0			&0				&\lambda_2
	\end{pmatrix}\in \mathbb{R}^{4\times 4},
	\end{align}
	and two output matrices such that $C^1$ satisfies the assumed structure and $C^2$ violates it, 
	\begin{align}
	C^1 &= 
	\begin{pmatrix}
	c_{1,1}^1 	&c_{2,1}^1 		&0			&0
	\end{pmatrix} \in \mathbb{R}^{1\times 4}, \\
	C^2 &= 
	\begin{pmatrix}
	0 	&c_{2,1}^2  	&0			&c_{2,2}^2
	\end{pmatrix} \in \mathbb{R}^{1\times 4},
	\end{align}      
	with $c_{1,1}^1, c_{2,1}^1, c_{2,1}^2$ and $c_{2,2}^2\neq 0$. Then, the reduced system equations for $l\in\lbrace 1,2\rbrace$ are respectively 
	\begin{align}
	A^1_* &= 
	\begin{pmatrix}
	\lambda_1 	&1 			\\
	0			&\lambda_1	\\
	\end{pmatrix}\in \mathbb{R}^{2\times 2},\\
	C^1_* &= 
	\begin{pmatrix}
	c_{1,1}^l 	&0 
	\end{pmatrix} \in \mathbb{R}^{1\times 2}.
	\end{align}
	and 
	\begin{align}
	A^2_* &= 
	\begin{pmatrix}
	\lambda_2 	&1 			\\
	0			&\lambda_2	\\
	\end{pmatrix}\in \mathbb{R}^{2\times 2},\\
	C^2_* &= 
	\begin{pmatrix}
	0 	&c_{2,2}^2 
	\end{pmatrix} \in \mathbb{R}^{1\times 2}.
	\end{align}
	
	According to Lemma \ref{lemma:ObseMatrixAC}, the modes $(x_1,x_2)$ fall into the observability space of sensor 1, whereas $(x_3,x_4)$ are unobservable by both sensor 1 and 2. 
\end{example}
W.l.o.g. we express (\ref{eq:outEqTrans}) as
\begin{align}
Y^l(k) &= \left( C^l_{\textbf{O}^l} \ \textbf{0} \right) 
\begin{pmatrix}
X_{\textbf{O}^l}(k)  \\ X_{\bar{\textbf{O}}^l}(k)
\end{pmatrix}, \ \text{for}\ l\in \lbrace 1,2\rbrace \label{eq:outEqTransDec}
\end{align}
where $X_{\textbf{O}^l}$ and $X_{\bar{\textbf{O}}^l}$ are the system states observable and unobservable at the $l$-th sensor, respectively. Let $\mathbb{X}_{\textbf{O}^l}$ and $\mathbb{X}_{\bar{\textbf{O}}^l}$ denote the subsets of observable and unobservable states at sensor $l\in\lbrace 1,2\rbrace$. Then, by virtue of the discussion after Lemma \ref{lemma:ObseMatrixAC}, the system modes can be partitioned into three distinct classes: 
\begin{itemize}
	\item modes common to both encoders $\gamma^1$ and $\gamma^2$, i.e., $X^0 = \mathbb{X}_{\textbf{O}^1} \cap \mathbb{X}_{\textbf{O}^2}$,
	\item modes observed \textit{only} by $\gamma^1$, i.e., $X^1 = \mathbb{X}_{\textbf{O}^1} \cap\mathbb{X}_{\bar{\textbf{O}}^2} $,
	\item modes observed \textit{only} by $\gamma^2$, i.e., $X^2 = \mathbb{X}_{\textbf{O}^2} \cap \mathbb{X}_{\bar{\textbf{O}}^1}$.
\end{itemize}


In order to estimate the different states, a dead-beat observer \cite{dorf2001} is incorporated at each sensor. By collecting the measurements $\left(Y^l(n(k-1)),\cdots,Y^l(nk)\right)$ for any $k\in\mathbb{Z}_{\geq 1}$ over a period of time $n$, the $l$-th sensor recovers its observable states and generates the following measurements 
\begin{align}
	\bar{Y}^l(nk) &= X_{\mathbf{O}^l}(nk), \ \text{for } l\in \lbrace 1,2\rbrace. \label{decOutput}    
\end{align}
Using the mode classes $X^0, X^1$, and $X^2$, it is hence possible to decompose the original system into three subsystems with dynamic matrices $A^{(0)},\ A^{(1)}$ and $A^{(2)}$ and reform (\ref{eq:stateEqTrans})-(\ref{eq:outEqTrans}) to obtain the setup depicted in Fig.~\ref{fig:setup} characterized by (\ref{originalState})-(\ref{originalOutput}). 



\section{Numerical Example: The Binary Adder Channel (BAC)}
\label{sec:numExample}

 \subsection{Zero-Error Capacity Region of the BAC}
 \label{sec:BAC}
A well-known example of a MAC in the literature is the \textit{binary adder channel} (BAC). This channel model consists of two independent users communicating with one receiver via a common discrete memoryless channel such that
\begin{align} \label{defBAC}
Y_k = X_k^1 + X_k^2 \in \left\lbrace 0, 1, 2 \right\rbrace, \ \forall k \in \mathbb{Z}_{\geq 0},
\end{align}
where $X_k^1, X_k^2 \in \lbrace 0,1 \rbrace$ are the inputs generated by users 1 and 2 respectively, and $Y_k$ is the corresponding channel output. The $i$-th user selects a particular \textit{codeword} $x^i_{1:n}$ from a finite set $\mathcal{X}^i \subseteq \left\lbrace 0, 1 \right\rbrace^n$, where $n \in \mathbb{Z}_{\geq 1}$ is the code block-length. At the receiver side, zero-error decoding is achieved when each element of the sumset 
\begin{align}
\mathcal{Y} = \left\lbrace x^1_{1:n}  + x^2_{1:n} : \forall x^1_{1:n} \in \mathcal{X}^1,\ x^2_{1:n} \in \mathcal{X}^2 \right\rbrace
\end{align}
appears exactly once, i.e., one-to-one correspondence, and hence unambiguous decoding is always possible. Note that the addition here is component-wise over $\mathbb{Z}$. No closed form of the BAC zero-error capacity region $\mathcal{C}_0$ has been found yet. However, outer bounds \cite{ahlswede71, liao72, austrin2018} and inner bounds \cite{mattas2005} are available. A zero-error code from \cite{mattas2005} is used in the numerical example of the next section. 
\subsection{Application to the State Estimation Problem (Theorem \ref{th:StateEstimation})}
We explore a setup where the states of two independent dynamical systems with given topological vector of entropies $h = (h_1,h_2)^T$ are measured, encoded and transmitted by users $\mcE^1$ and $\mcE^2$ through a BAC (\ref{defBAC}) and eventually recovered at the receiver end. 
\paragraph{Simulation Details}
We implement two scalar LTI systems described by the following equations
\begin{subequations}
	\begin{align}
	x^i(k+1) &= a^i x^i(k) + v^i(k), \label{eq:simSystemState} \\
	y^i(k) &= x^i(k),\ \ i\in\lbrace 1,2\rbrace,\label{eq:simSystemOut}  
	\end{align} 
\end{subequations}
where $a^i = 2^{h_i}$ and $v^i(k)$ denotes the process noise with range $[-1,1]$. At time instant $k=0$, the systems' initial states $\lbrace x^i(0) \rbrace_{i=1}^{2}$ and noise term $\lbrace v^i(0)\rbrace_{i=1}^{2}$ are uniformly chosen from $[-1,1]$. At each sampling time $nk$ the encoder produces an estimate $\hat{x}^i(nk)$ of the true state $x^i(nk)$. Then, the error $e^i(nk) := y^i(nk) - \hat{x}^i(nk)$ is processed by means of an adaptive uniform quantizer $\mcQ^i$ as follows. Firstly, the interval $[-1,1]$ is partitioned into $K_i$ sub-intervals of equal size whose respective midpoints are denoted by $\omega(1),\cdots,\omega(K_i)$. Next, the error $e^i(nk)$ is scaled and quantized via 
\begin{align}
e_q^i(nk) := \mcQ^i\left(e^i(nk)/\ell_k^i\right) \in \lbrace \omega(1),\cdots,\omega(K_i) \rbrace.  
\end{align}        
Both encoder and decoder agree on the initial value of the scaling factor $\ell^i_0$ and update it at each instant $nk$ as follows 
\begin{align}
\ell^i_{k+1} := \frac{(a^i)^n}{K_i}\ell^i_k + \Psi^i_{\max},
\end{align}
where the term $\Psi^i_{\max} = \sum_{\xi =0}^{n-1}(a^i)^{n-1-\xi}$ forms an upper bound on the accumulated noise $\Psi^i_{n-1}(k)$ during a time window of duration $n$. The quantized error $e_q^i(nk)$ is then encoded by $\mcE^i$ which employs a zero-error codebook with block-length $n$ and a rate $ R^i = \log_2(K_i)/n$. The decoder $\mcD$ uses the received codewords to generate the $e_q^i(nk)$ with exactly zero decoding errors. Finally, the state estimate $\hat{x}^i(n(k+1))$ is updated by 
\begin{align}
\hat{x}^i(n(k+1)) = (a^i)^n \left(\hat{x}^i(nk) + \ell^i_k e_q^i(nk)\right).
\end{align}
Note that $\hat{x}^i(k^{'})$ for $k^{'} \in (nk,n(k+1))$ is computed via 
\begin{align}
\hat{x}^i(k^{'}) = (a^i)^{k^{'}-nk} \hat{x}^i(nk).
\end{align} 
\paragraph{Numerical Results \& Discussion}  
The system (\ref{eq:simSystemState})-(\ref{eq:simSystemOut}) is simulated for different parameter settings as outlined in Table \ref{tab:menas}. Fig.~\ref{fig:stateRealisation} depicts a realization of the systems' unstable states, i.e., $X^1=x^1$ and $X^2=x^2$, on both linear and logarithmic scales. Zero-error communication over the BAC is accomplished by means of the following binary codebooks \cite{mattas2005} of block-length $n = 6$: 
\begin{align}
	\mcX^1 &= \lbrace 0, 3, 6, 15, 17, 27, 36, 46, 48, 57, 60, 63\rbrace, \nonumber \\
	\mcX^2 &= \lbrace 8, 9, 10, 13, 21, 22, 26, 28, 29, 30, 33, 34, 35, 37, 41, \nonumber\\
			  &\ \ \ \ \  42, 50, 53, 54, 55\rbrace, \nonumber
\end{align}
where $\mcX^1$ and $\mcX^2$ are expressed in the decimal basis. 
In this experiment, $10^5$ realizations of system (\ref{eq:simSystemState})-(\ref{eq:simSystemOut}) with different initial conditions and noise values randomly chosen from the range $[-1,1]$ are simulated. At each time step $k$, the estimation error with the maximum norm over all realizations is selected and plotted in Fig.~\ref{fig:scenario1}-\ref{fig:scenario2} for both scenarios 1 and 2. Note that the topological entropies in Scenario 1 are selected such that they lie within the inner region \cite{mattas2005} of the BAC's $\mcC_0$. This guarantees that $h$ is within $\mcC_0$. On the other hand, for Scenario 2, the topological entropies are chosen from beyond the outer bound \cite{austrin2018} on $\mcC_0$.

When the vector of topological entropies $h$ lies within the BAC's $\mcC_0$, it is possible to find code rates such that $R^i > h_i$ for $i \in \lbrace 1,2 \rbrace$. Thus, the estimation error is bounded in accordance with Theorem \ref{th:StateEstimation} as exhibited in Fig.~\ref{fig:scenario1}. Scenario 2 on the other hand illustrates the case where $h \notin \mathrm{int}(\mcC_0)$, and hence, it becomes impossible to reconstruct the state estimates at the receiver with zero error. One can see in Fig.~\ref{fig:scenario2} that the estimation error grows exponentially with time.   
\begin{table}[htp]
	\centering
	\caption{Simulation Parameters. \textit{The code block-length is set to $n=6$}.}
	\begin{tabular}{l l l }
		\hline 
		\textbf{Scenario N\textdegree} & \textbf{1} & \textbf{2} \\ \hline\hline
		$h_1$  			&0.15   &0.85   		\\
		$R^1$ 	 		&0.597  &0.597  		\\
		$h_2$  			&0.18   &0.95  		\\
		$R^2$  			&0.720  &0.720			\\
		$\llb V^i\rrb_{i=1,2}$  &$[-1,1]$ &$[-1,1]$ \\
		\hline 
	\end{tabular}
	\label{tab:menas}
\end{table}
\begin{figure}[htp]
	\centering
	\includegraphics[scale=0.38]{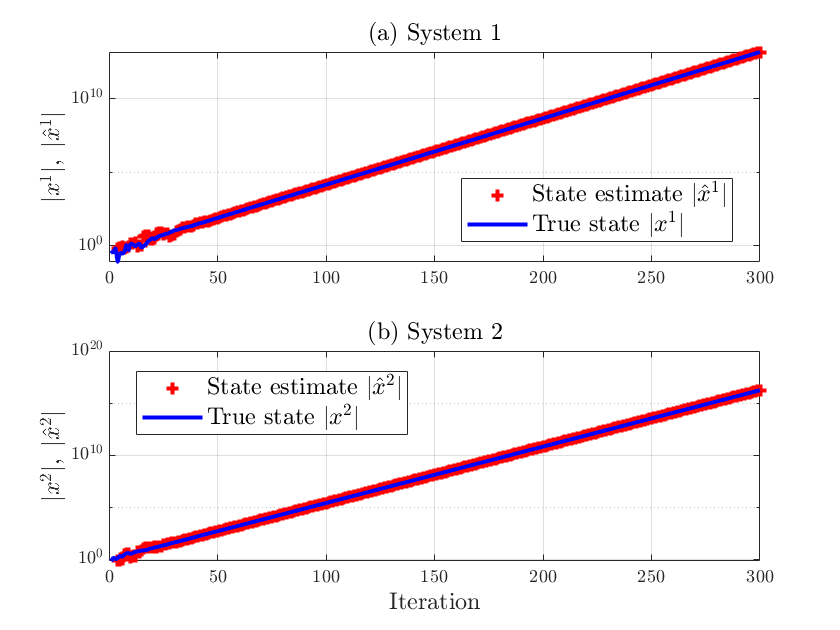}
	\caption{Example of state realizations $x^1, x^2$ and their corresponding estimations $\hat{x}^1, \hat{x}^2$ for unstable systems 1 and 2 on a \textit{logarithmic} scale.}
	\label{fig:stateRealisation}
\end{figure}

\begin{figure}[htp]
	\centering
	\includegraphics[scale=0.38]{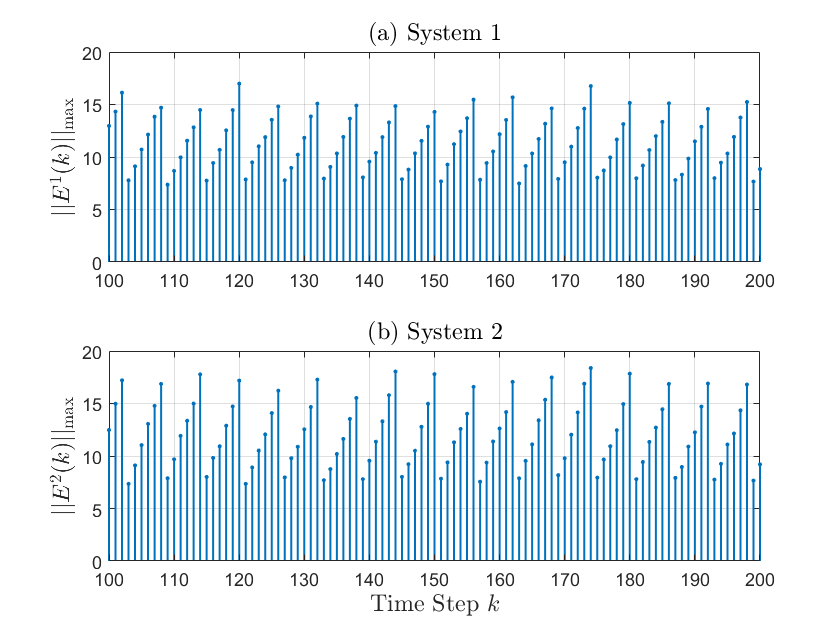}
	\caption{Empirical maximum error norms on a \textit{linear} scale for Scenario 1. The topological entropy vector $h \in \mathrm{int}(\mcC_0)$, and the code rates $R^1 > h_1$ and $R^2 > h_2$.}
	\label{fig:scenario1}
\end{figure}

\begin{figure}[htp]
	\centering
	\includegraphics[scale=0.38]{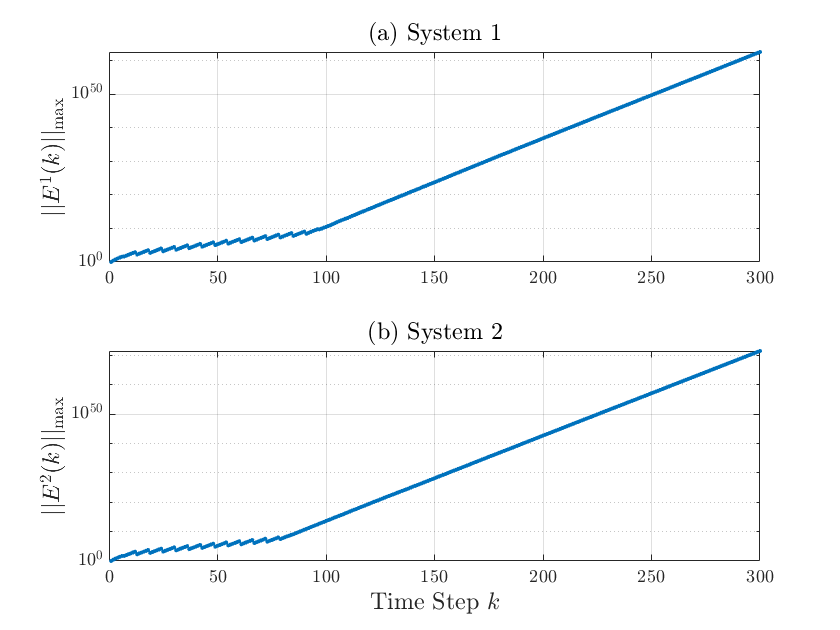}
	\caption{Empirical maximum error norms on a \textit{logarithmic} scale for Scenario 2. The topological entropy vector $h \notin \mathrm{int}(\mcC_0)$, and the zero-error code rates $R^1 < h_1$ and $R^2 < h_2$.}
	\label{fig:scenario2}
\end{figure}
\section{Conclusion}
\label{sec:conclusion}

This article represents a preliminary step in the direction of developing an in-depth understanding of the fundamental connection between the different components of a networked system in the context of distributed state estimation. Motivated by the relevance of zero-error capacity as figure of merit to assess system performance in worst-case scenarios, we derived a multi-letter characterization of the zero-error capacity region of an $M$-user MAC with common message. Unlike classical information theory, the tools from its nonstochastic analogue allowed us to obtain a result which is not only valid for asymptotically large coding block-lengths, but also true for finite ones. This MAC structure models a relevant scenario where $M$ sensors observe a different subset of the overall system's dynamical modes, and the common message represents the latent correlation between their readings, i.e., the overlap between the observed subsets.

Next, the problem of distributed state estimation across a two-input, single-output MAC with bounded noise was studied. We provided tight necessity and sufficiency conditions to ensure uniformly bounded errors at the receiver end. It has been shown that in order to achieve this goal, the vector of the plants' topological entropies must lie within the zero-error capacity region of the communication link. This result establishes a connection between the intrinsic properties of the linear systems and the channel characteristics.           

Some of the remaining research questions that we need to address in future work include the following.
\begin{itemize}
	\item The problem we considered in this article is distributed state estimation. It is also interesting if the results obtained here are appropriately extended to cover the stabilizability of LTI systems across a shared MAC. To this end, a characterization of the zero-error capacity region of MAC with feedback in the nonstochastic framework has to be derived. 
	
	\item The coding scheme using Luenberger observers, which was constructed in this article, allowed us to establish
	the achievability argument to prove the main result. Nonetheless, in order to achieve a good estimation performance in practice, more	sophisticated coding schemes, e.g., using zooming techniques \cite{brockett2000}, must be developed.
	
	\item As outlined in the course of the paper, an exact $\mcC_{0}$ characterization for many MACs, e.g., BAC, remains unknown. Hence, it is also of interest to develop low-cost algorithms to compute nonstochastic information and to ultimately provide the specific zero-error capacity region of any desired MAC, or at least to reduce the gap between the upper and lower bounds characterizing the region.      
\end{itemize}

\appendices 

\section{Proof of Theorem \ref{thm:convexity}}
\label{sec:proofConvexity}
Select two achievable rate triples $R^{'},\ R^{''} \in \mcC_0$ (\ref{eq:zeCap}). 
For any $ \epsilon  > 0$, $\exists (n,R^{'}_n)$ and $(m,R^{''}_m)$ such that
\begin{align}
	\Vert R^{'} - R^{'}_n   \Vert \leq \epsilon, \
	\Vert R^{''} - R^{''}_m	\Vert \leq \epsilon ,\label{rates}
\end{align} 
where $R^{'}_n \in \mcC_{0,n}$ and $R^{''}_m \in \mcC_{0,m}$ for sufficiently large $n, m \in \mathbb{Z}_{\geq 1}$ that denote the block-lengths of the zero-error codes operating at $R^{'}_n$ and $R^{''}_m$. First, we show that for any $\alpha \in (0,1)$, we can construct an achievable rate triple 
\begin{align}
R = \alpha R^{'} + (1 - \alpha) R^{''}. 
\label{eq:kharyet23062022@10:35am}
\end{align} 
We define the terms $e^{'}_n$ and $e^{''}_m$ as follows
\begin{align}
	e^{'}_n  :=  R^{'} - R^{'}_n ,\ e^{''}_m :=  R^{''} - R^{''}_m. \label{ineq:epsilonCon}
\end{align}
By transmitting $j$ blocks of length $n$ at rate $R^{'}_n$, followed by the remaining $k$ blocks of length $m$ at rate $R^{''}_m$, we obtain 
\begin{align}
\bar{R}_{n,m} 
&= \frac{jn}{jn + km} \left( R^{'} - e^{'}_n  \right) + \frac{km}{jn + km} \left( R^{''} - e^{''}_m \right) \label{rateSplitted}.
\end{align}   
For sufficiently small $\delta>0$ and given integers $n,m \geq 1$, we seek integers $j,k$ such that 
\begin{align}
\left\vert \alpha - \frac{jn}{jn + km} \right\vert \leq \delta.
\label{ineq:kharyafaddetni26062022@10.50am}
\end{align}
Thus, for a fixed $\alpha \in (0,1)$, the ratio $(k/j) \in \mathbb{Q}$ must satisfy 
\begin{align}
\frac{n}{m} \left[ \frac{1}{\alpha + \delta}  - 1\right]  \leq \frac{k}{j} \leq \frac{n}{m} \left[\frac{1}{\alpha -\delta}- 1\right] .
\label{ineq:kharyafaddetni12062022@6.05pm}
\end{align}
For arbitrarily small $\delta$ and $\alpha \in (0,1)$, there exist $k,j \in \mathbb{Z}_{\geq 1}$ such that (\ref{ineq:kharyafaddetni12062022@6.05pm}) holds. Then, we have
\begin{align}
&\Vert \bar{R}_{n,m} - R \Vert = \nonumber \\
&= \left\Vert \frac{jn}{jn + km} \left( R^{'} - e^{'}_n  \right) + \frac{km}{jn + km} \left( R^{''} - e^{''}_m \right) - \alpha R^{'} - \right. \nonumber \\ 
&\ \ \ \left. (1-\alpha) R^{''} \right\Vert \nonumber  \\
&\leq \left\vert \frac{jn}{jn + km} -\alpha \right\vert \Vert R^{'}\Vert + \left\vert - \frac{jn}{jn + km} + \alpha \right\vert \Vert R^{''} \Vert + \nonumber \\
&\ \ \ \frac{jn}{jn + km} \Vert e^{'}_n \Vert  +  \frac{km}{jn + km} \Vert e^{''}_m \Vert  \label{ineqConvex:step1}  \\ 
&\leq \delta (\Vert R^{'}\Vert + \Vert R^{''}\Vert ) + \left(\frac{jn}{jn + km} + \frac{km}{jn + km} \right)  \epsilon \label{ineqConvex:step2}  \\
&= \delta (\Vert R^{'}\Vert + \Vert R^{''}\Vert) +   \epsilon, \nonumber
\end{align}
where (\ref{ineqConvex:step1}) follows from the triangle inequality, and (\ref{ineqConvex:step2}) holds by virtue of (\ref{ineq:kharyafaddetni26062022@10.50am}) and (\ref{rates}).    
Hence, by making the block-lengths $n$ and $m$ sufficiently large, therefore, for arbitrarily small $\epsilon, \delta$, the difference $\Vert \bar{R}_{n,m} - R \Vert$ can be made arbitrarily small, i.e., the rate $R$ is arbitrarily close to the zero-error code rate $\bar{R}_{n,m}$. We then deduce that $R \in \mcC_{0}$, which proves the desired property.

\section{Proof of (\ref{eq:NecessityIneq})}
\label{sec:proofIneq}
The proof below is similar to the technique used in \cite{nair2013} and \cite{saberi2020} in the context of point-to-point channels but as the setup here is different, we include it for means of completeness. As mentioned in Section \ref{subsec:necessityProof}, in the following proof we only consider the unstable system eigenvalues without loss of generality. Furthermore, it follows from Def.~\ref{def:unifBoundError} that uniformly bounded errors are obtained for any uniformly bounded noise ranges $\llb V(k)\rrb$ and $\llb W(k)\rrb$ as well as any initial state range $\llb X(0)\rrb$ contained in some $l$-ball $\mathbf{B}_l$. Hence, for the upcoming analysis we set the noise terms to zero, i.e., $\llb V(k)\rrb = \llb W(k)\rrb = \lbrace 0\rbrace$ and construct the range of initial states in the following manner. First, select $\epsilon \in \left( 0, 1 - \max_{\ell: \vert \lambda^i_\ell \vert > 1} \frac{1}{\vert \lambda^i_\ell \vert} \right) $. Note that this interval is nonempty by virtue of Assumption \textbf{(A5)}. We then divide the interval $\left[ -l,l  \right] $ on the $\ell$-th axis into $\kappa_\ell$ equal subintervals of length $2l/\kappa_\ell$ such that 
\begin{align}
\kappa_\ell := \left\lfloor \left\vert \left(1 - \epsilon \right)\lambda_\ell \right\vert^k \right\rfloor, \ \ \ell \in [1:d_i] . 
\end{align} 
Let $p_\ell(s)$ denote the midpoints of the subintervals, for $s = [1:\kappa_\ell]$ and construct an interval \textbf{I}$_\ell(s)$ centered at $p_\ell(s)$ such that its length is equal to $l/\kappa_\ell$. We define the family of hypercuboids $\mathcal{H}$ as follows 
\begin{align}
\mathcal{H} = \left\lbrace  \left( \prod_{\ell = 1}^{d_i} \textbf{I}_\ell (s) \right): s \in [ 1 : \kappa_\ell ],\  \ell \in [1:d_i] \right\rbrace.
\label{eq=H} 
\end{align}   
Observe that any two hypercuboids from $\mathcal{H}$ are separated by a distance of $l/\kappa_\ell$ along the $\ell$-th axis for each $\ell \in [1:d_i]$.

In the following analysis the superscript $i$ referring to the respective plant is omitted unless otherwise stated. The initial state range is set as $\llb X^i(0) \rrb = \bigcup_{\textbf{H} \in \mathcal{H}} \textbf{H} \subset \textbf{B}_l (0) $. The operator $\text{dm}(\cdot )$ is defined as the diameter of a set using $l_\infty$-norm. Then, as $\llb E^i_j(k) \rrb \supseteq \llb E^i_j(k) \vert q(0:k-1) \rrb $, 
\begin{align}
\text{dm}(\llb E^i_j(k) \rrb) &\geq \text{dm}\left( \llb E^i_j(k) \vert q(0:k-1) \rrb \right)  \label{eq=a} \\
&= \text{dm}\left( \llb X^i_j(k) - \delta^i_j \left( k, q \left( 0:k-1 \right) \right)  \vert q(0:k-1) \rrb \right) \nonumber \\
&= \text{dm}\left( \llb (A_j^i)^k X^i_j(0) \vert q(0:k-1) \rrb \right)   \label{eq=b} \\
&= \sup_{r,p \in \llb X^i_j(0) \vert q(0:k-1) \rrb} \Vert (A_j^i)^k (r - p) \Vert  \nonumber \\
&\geq \sup_{r,p \in \llb X^i_j(0) \vert q(0:k-1) \rrb} \frac{\Vert (A_j^i)^k (r - p) \Vert_2 }{\sqrt{d_i}} \nonumber \\
&\geq \sup_{r,p \in \llb X^i_j(0) \vert q(0:k-1) \rrb} \frac{\sigma_{\min} \left( \left(A_j^i\right)^k\right) \Vert r - p \Vert_2 }{\sqrt{d_i}} \nonumber \\
&\geq  \sigma_{\min} \left( \left(A_j^i\right)^k\right) \frac{ \text{dm}\left(\llb X^i_j(0)  \vert q(0:k-1) \rrb\right) }{\sqrt{d_i}}, \nonumber	                        	                              
\end{align}
where $k \in \mathbb{Z}_{\geq 0},\ q(0:k-1) \in \llb Q(0:k-1 ) \rrb$, $\Vert \cdot \Vert_2 $ denotes the Euclidean norm and $\sigma_{\min} ( \cdot )$ refers to the smallest singular value. The inequality (\ref{eq=a}) holds because conditioning reduces the range, and (\ref{eq=b}) is valid since the diameter of a uv range is translation invariant. 
Now, note that the Yamamoto identity \cite{yamamoto67} states that 
\begin{align}
\lim_{k \rightarrow \infty} \left(  \sigma_{\min} ( (A_j^i)^k) \right)^{1/k} = \left\vert \lambda_{\min}\left(A_j^i\right) \right\vert,
\end{align}   
with $\lambda_{\min}$ being the eigenvalue with the smallest magnitude. Then, since $j \in [1: m]  < \infty $, i.e., there are finitely many real Jordan blocks $A_j^i$, there exists $k_\epsilon \in \mathbb{Z}_{\geq 0}$ such that 
\begin{align}
\sigma_{\min} \left( (A_j^i)^k  \right) \geq \left( 1 -\frac{\epsilon}{2}  \right)^k \vert \lambda_{\min} ( A_j^i ) \vert ^k, \text{ for } k \geq k_\epsilon. 
\end{align}
Additionally, by \textit{uniform boundedness of errors} (\ref{eq:uniformBoundedness}) there exists $\xi >0$ such that
\begin{align}
\xi &\geq  \sup \llb \Vert E_j^i(k) \Vert \rrb \geq \frac{1}{2} \text{dm} \left( \llb E^i_j(k) \rrb \right) \nonumber \\
&\geq \left\vert \left( 1 -\frac{\epsilon}{2}  \right)  \lambda_{\min} ( A_j^i )  \right\vert  ^k \frac{ \text{dm} \left( \llb X_j^i(0) \vert q(0:k-1) \rrb \right)}{2 \sqrt{d_i}}.
\end{align} 
For large enough $k \in \mathbb{N}$, the hypercuboid family $\mathcal{H}$ in (\ref{eq=H}) is an $\llb X^i(0) \vert Q(0:k-1 )\rrb$-overlap isolated partition of $\llb X^i(0) \rrb$. To show this, we suppose in contradiction that $\exists \textbf{H} \in \mathcal{H} $ that is overlap connected in the family $\llb X^i(0) \vert Q(0:k-1 ) \rrb$ with another hypercuboid from $\mathcal{H}$. Therefore, there would exist a set $ \llb X^i(0) \vert q(0:k-1) \rrb$ which contains a point $r_j \in \textbf{H}$ and a point $p_j \in \textbf{H}^{'}  $, with $\textbf{H}^{'} \in \mathcal{H} \setminus \textbf{H} $ such that $r_j$ and $p_j$ are overlap connected. This implies 
\begin{align}
\Vert p_j - r_j \Vert &\leq \text{dm}\left(\llb X^i(0) \vert q(0:k-1 \rrb\right) \nonumber \\ 
&\leq \frac{2\sqrt{d_i} \xi}{\left\vert \left( 1 -\frac{\epsilon}{2}  \right)  \lambda_{\min} ( A^i_j )  \right\vert  ^k},
\label{ineq:contradiction1}                  
\end{align}
for $j \in [1:m]$ and $k \geq k_\epsilon$. Nonetheless, note that by construction the distance between any two hypercuboids in $\mathcal{H}$ is equal to $l/\kappa_\ell$ along the $\ell_i$-th axis. Therefore,
\begin{align}
\Vert p_j - r_j \Vert \geq \frac{l}{\kappa_\ell} 
&\leq    \frac{l}{(1 - \epsilon )^k \vert \lambda_{\min}(A_j^i) \vert^k } 
\label{ineq:contradiction2}
\end{align} 
For sufficiently large $k$, it is possible to obtain $\left( \left( 1 - \epsilon/2 \right) /  \left( 1 - \epsilon \right) \right)^k > 2 \sqrt{d_i} \xi /l $. Hence, the RHS of (\ref{ineq:contradiction2}) would exceed the RHS of (\ref{ineq:contradiction1}) resulting in a contradiction. Thus, when $k$ is large enough, the family $\mathcal{H}$ is $\llb X^i(0) \vert Q(0:k-1) \rrb$-overlap isolated partition of $\llb X^i(0) \rrb$. As the cardinality of any $\llb X^i(0) \vert Q(0:k-1) \rrb$-overlap isolated partition is upper bounded by the maximin information $\left\vert \llb X^i(0) \vert Q(0:k-1) \rrb_* \right\vert$, we obtain 
\begin{align}
I_*[X^i(0);&Q(0:k-1)] = \log \vert \llb X^i(0) \vert Q(0:k-1) \rrb_* \vert \nonumber \\
&\geq \log \vert \mathcal{H} \vert 	               \nonumber
= \log \left( \prod_{\ell=1}^{d_i} \left\lfloor \vert (1 - \epsilon )\lambda^i_\ell \vert^k  \right\rfloor \right)  \nonumber \\
&\geq \log \left( \prod_{\ell=1}^{d_i} 0.5 \left\vert (1 - \epsilon )\lambda^i_\ell \right\vert^k  \right)  \nonumber \\
&= k \left( d_i \log (1 - \epsilon ) - \frac{d_i}{k} + \sum_{\ell=0}^{d_i} \log \left\vert \lambda^i_\ell \right\vert \right).
\label{ineq:upperbound} 
\end{align}
Hence, for $i \in \left\lbrace 0,1,2 \right\rbrace$ it follows from (\ref{ineq:upperbound})
\begin{align}
\frac{I_*[X^i(0); Q(0:k - 1)]}{k} 
&\geq d_i \log (1 - \epsilon ) - \frac{d_i}{k} + \sum_{\ell=0}^{d_i} \log \vert \lambda^i_\ell \vert.
\nonumber 
\end{align}
Using the monotonicity of $I_*$ and then letting $k\to\infty$ and $\eps\to 0$, we obtain
\begin{align}
\frac{I_*[X^i(0:k -1); Q(0:k - 1)]}{k} &\geq  \frac{I_*[X^i(0); Q(0:k - 1)]}{k} \nonumber \\
&\geq \sum_{\ell=0}^{d_i} \log \vert \lambda^i_{\ell} \vert. \hspace{12mm} \Box  \nonumber
\end{align}
\vspace*{-5mm}
\section{Proof of Lemma \ref{lemma:conditionalIstar}}
\label{sec:appLemma}
To prove Lemma \ref{lemma:conditionalIstar} we use the common variable interpretation of (conditional) nonstochastic information as introduced in (\ref{defIstar}) and (\ref{defcondinfo}), respectively. Consider three uv’s $\Lambda , \Theta$ and $\Omega$ such that $\Lambda$ and $\Theta$ are unrelated. Let $\mcZ_{\Theta}$ denote the set of all uv’s $Z_{\Theta}$ such that $Z_{\Theta} \perp \Theta$ and $Z_{\Theta} \equiv f(\Lambda , \Theta) = g(\Omega , \Theta)$. Thus, 
\begin{align}
I_*[\Lambda ; \Omega \vert \Theta] = \max_{Z_{\Theta} \in \mcZ_\Theta} \log \vert \llb Z_\Theta \rrb \vert.
\end{align}
Additionally, the set $\mcZ$ consists of all uv's $Z$ such that $Z \equiv \phi (\Lambda) = \psi (\Omega) $, and hence, 
\begin{align}
I_*[\Lambda ; \Omega] = \max_{Z \in \mcZ} \log \vert \llb Z \rrb \vert. 
\end{align}
Since $\Lambda \perp \Theta$, then $\phi (\Lambda) \perp \Theta$ and subsequently $Z \perp \Theta$. Therefore, $\mcZ \subseteq \mcZ_\Theta $, and thus 
\begin{align}
\log \vert \llb Z \rrb \vert \leq \log \vert \llb Z_\Theta \rrb \vert. 
\label{ineq:cv1}
\end{align} 
By maximizing both LHS and RHS of (\ref{ineq:cv1}) over $\mcZ$ and $\mcZ_\Theta$, we obtain 
\begin{align}
I_*[\Lambda ; \Omega] \leq I_*[\Lambda ; \Omega \vert \Theta].\hspace{15mm} \Box \nonumber
\end{align}
\section{Proof of Lemma \ref{lemma:ObseMatrixAC}}
\label{sec:LemmaObs}
The following proof is based on Hautus test \cite[p.~156]{linSysChen} that states the following. 
\begin{theorem}
	A system $(A,C)$ is observable if and only if the matrix $\begin{pmatrix} A-\lambda I \\ C \end{pmatrix}$ has full column rank at every eigenvalue $\lambda$ of $A$. 
\end{theorem}

Let $A\in\mathbb{R}^{d\times d}$ be the state matrix in the sense of (\ref{eq:AmatrixJordan}) and $C^l\in\mathbb{R}^{b_l\times d}$ be the output matrix as defined in (\ref{eq:Cmatrix}). For $s\in\mathbb{R}$, consider the matrix $\begin{pmatrix} A- s I \\C^l\end{pmatrix}$ with the following form 
\begin{align}
\resizebox{1\hsize}{!}{%
	$\left(\begin{array}{c c c c |c c c c | c | c c c c} 
	\lambda_1-s&1 & & & & & & & & & &  \\
	&\ddots&\ddots & & & & & & & & &  \\
	&			&\lambda_1-s&1 & & & & & & & \\	
	&			&			&\lambda_1-s & & & & & & & \\
	\hline
	&			& & &\lambda_2-s & 1 & & &  \\
	&			& &	& &\ddots & \ddots & &  \\
	&			& & & & &\lambda_2-s & 1 & & &  \\
	&			& & & & & &\lambda_2-s & & & \\
	\hline 
	&			& & & & & &  & \ddots	& &	\\
	\hline 
	&			&	& &	& & & & &\lambda_{\Gamma}-s & 1		\\
	&			&	&	& & & & & & &\ddots &\ddots\\
	&			&	& &	& & & & & & &\lambda_{\Gamma}-s & 1		\\
	&			&	& & & & & & & & &  &\lambda_{\Gamma}-s \\	
	\hline	
	c^l_{1,1} &c^l_{2,1} &\cdots &c^l_{d_{11},1} &c^l_{1,2} &c^l_{2,2} &\cdots &c^l_{d_{21},2}&\cdots  & c^l_{1,\Gamma} & c^l_{2,\Gamma} & \cdots  & c^l_{d_{\Gamma\ell_\Gamma},\Gamma}
	\end{array}\right),$
}%
\label{eq:matrixAsI} 
\end{align}
where the empty entries are zeros. In the following analysis, we focus on the first Jordan block without loss of generality. Hence, for $s=\lambda_1$, the matrix (\ref{eq:matrixAsI}) consists of columns whose components are either 
\begin{itemize}
	\item All zeros except for the last $b_l$ elements (the respective column of $C^l$). This is the leading column of the Jordan block associated with $\lambda_1$; or
	\item All zeros except for one non-zero element (either $1$ or $\lambda_j-\lambda_1$ for $j\in[2:\Gamma]$) and the last $b^l$ elements (the respective column of $C^l$). This is the leading column of each of the Jordan blocks associated with $\lambda_j$ for all $j\in[2:\Gamma]$; or
	\item All zeros except for two non-zero elements ($1$ and $\lambda_j-\lambda_1$ for $j\in[2:\Gamma]$) and the last $b^l$ elements (the respective column of $C^l$). 
\end{itemize}  
As the rank of (\ref{eq:matrixAsI}) will not change by elementary row operations, by means of Gaussian elimination we transform it into the following form
\begin{align}
\resizebox{1\hsize}{!}{%
	$\left(\begin{array}{c c c c |c c c c | c | c c c c} 
	0&1 & & & & & & & & & &  \\
	&\ddots&\ddots & & & & & & & & &  \\
	&			&0&1 & & & & & & & \\	
	&			&			&0 & & & & & & & \\
	\hline
	&			& & &\lambda_2-\lambda_1 & 0 & & &  \\
	&			& &	& &\ddots & \ddots & &  \\
	&			& & & & &\lambda_2-\lambda_1 & 0 & & &  \\
	&			& & & & & &\lambda_2-\lambda_1 & & & \\
	\hline 
	&			& & & & & &  & \ddots	& &	\\
	\hline 
	&			&	& &	& & & & &\lambda_{\Gamma}-\lambda_1 & 0		\\
	&			&	&	& & & & & & &\ddots &\ddots\\
	&			&	& &	& & & & & & &\lambda_{\Gamma}-\lambda_1 & 0		\\
	&			&	& & & & & & & & &  &\lambda_{\Gamma}-\lambda_1 \\	
	\hline	
	c^l_{1,1} &\textbf{0} &\cdots &\textbf{0} &c^l_{1,2} &\textbf{0} &\cdots &\textbf{0}&\cdots  & c^l_{1,\Gamma} & \textbf{0} & \cdots  & \textbf{0}
	\end{array}\right),$
}%
\label{eq:matrixAsI_lambda1GaussianElim} 
\end{align}
where $\textbf{0}$ is a $b_l$-dimensional all-zero column vector. Clearly, the first column of (\ref{eq:matrixAsI_lambda1GaussianElim}) consists of $d$ elements that are zeros followed by $b_l$ entries corresponding to $c_{1,1}^l$. Therefore, to obtain a  full-ranked matrix (\ref{eq:matrixAsI_lambda1GaussianElim}), the column $c_{1,1}^l$ must be non-zero. Using the same argument with the remaining $\Gamma-1$ eigenvalues, Lemma \ref{lemma:ObseMatrixAC} is established.
\vspace*{-10mm}
\begin{IEEEbiography}[
	{\includegraphics[width=1in,height=1.2in,clip,keepaspectratio]{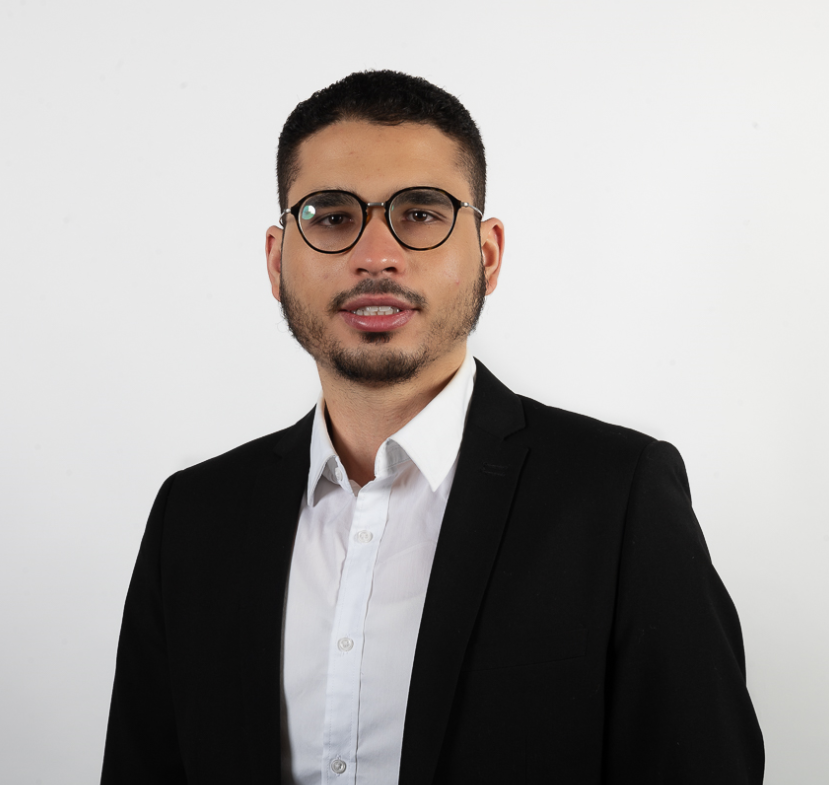}}]%
	{Ghassen Zafzouf}(S'19) 
	received a B.Sc. (2015) and M.Sc. (2018) in Electrical Eng. and Information Technology from the Technical University of Munich (Germany), an M.Eng. (2017) in Electrical Engineering from the University of Queensland (Australia), and a Ph.D. (2022) in Electrical Engineering from the University of Melbourne (Australia). His main research interests include information theory, digital communications and networked control systems. He has been the recipient of several prizes and awards including Germany's National Scholarship (2018) and the Melbourne Research Scholarship (2019). 
\end{IEEEbiography}
\vspace*{-10mm}
\begin{IEEEbiographynophoto}
	{\textbf{Girish N. Nair}}
	(Fellow, IEEE) was born in Malaysia  and obtained a PhD in electrical engineering from The University of Melbourne, Australia in 2000. He is currently a Professor with the Department of Electrical and Electronic Engineering, The University of Melbourne. Prof. Nair was the recipient of several prizes, including the CSS Axelby Outstanding Paper Award in 2014 and a SIAM Outstanding Paper Prize in 2006. From 2015 to 2019, he was an ARC Future Fellow, and since 2019 he has been the Principal Australian Investigator of an AUSMURI project.
\end{IEEEbiographynophoto}
\vspace*{-10mm}
\begin{IEEEbiographynophoto}
	{\textbf{Farhad Farokhi}}
	(Senior Member, IEEE) received the Ph.D. degree in automatic control from the KTH Royal Institute of Technology, Stockholm, Sweden, in 2014. He joined The University of Melbourne, Australia, where he is currently a Senior Lecturer (equivalent to Assistant/Associate Professor in North America). From 2018 to 2020, he was also a Research Scientist with the CSIRO's Data61, Canberra ACT, Australia. He has been involved in multiple projects on data privacy and cyber-security funded by the Australian Research Council, the Defence Science and Technology Group, the Department of the Prime Minister and Cabinet, the Department of Environment and Energy, and the CSIRO. Dr. Farokhi was the recipient of the VESKI Victoria Fellowship from the Victoria State Government, Australia, and the McKenzie Fellowship, the 2015 Early Career Researcher Award, and MSE Excellence Award for Early Career Research from The University of Melbourne. He is the Associate Editor for \textit{IET Smart Grid}, \textit{Results in Control and Optimization}, and Conference Editorial Board of IEEE Control System Society.
\end{IEEEbiographynophoto}


\end{document}